\documentclass[a4paper,UKenglish,cleveref, autoref, thm-restate]{lipics-v2021}

\title{Singly Exponential Translation of Alternating Weak \buchi Automata to Unambiguous \buchi Automata} 

\titlerunning{Singly exponential translation of AWAs to UBAs} 

\author{Yong Li}{University of Liverpool, UK}{liyong@liverpool.ac.uk}{https://orcid.org/0000-0002-7301-9234}{}

\author{Sven Schewe}{University of Liverpool, UK}{svens@liverpool.ac.uk}{https://orcid.org/0000-0002-9093-9518}{}

\author{Moshe Y. Vardi}{Rice University, USA}{vardi@cs.rice.edu}{https://orcid.org/0000-0002-0661-5773}{}

\authorrunning{Y. Li, S. Schewe, M. Vardi} 

\Copyright{Yong Li, Sven Schewe and Moshe Y. Vardi} 



\category{} 

\relatedversion{} 

\EventEditors{John Q. Open and Joan R. Access}
\EventNoEds{2}
\EventLongTitle{42nd Conference on Very Important Topics (CVIT 2016)}
\EventShortTitle{CVIT 2016}
\EventAcronym{CVIT}
\EventYear{2016}
\EventDate{December 24--27, 2016}
\EventLocation{Little Whinging, United Kingdom}
\EventLogo{}
\SeriesVolume{42}
\ArticleNo{23}

\usepackage{cite}
\usepackage{paralist}
\usepackage{mathtools}
\usepackage{xspace}
\usepackage{amssymb}
\usepackage{amsmath}
\usepackage{amsthm}
\usepackage{thmtools}
\usepackage{thm-restate}
\usepackage{wrapfig}
\usepackage{bm}
\usepackage{hyperref}
\usepackage[dvipsnames]{xcolor}

\nolinenumbers

\newcommand{\A}{\mathcal{A}}
\newcommand{\B}{\mathcal{B}}

\newcommand{\U}{\mathcal{U}}

\newcommand{\dual}[1]{\widehat{#1}}

\newcommand{\R}{\mathtt{R}}

\newcommand{\N}{\mathcal{N}}
\newcommand{\po}{\mathcal{P}}

\newcommand{\G}{\mathcal{G}}
\newcommand{\bigO}{\mathcal{O}}
\newcommand{\acctrue}{\mathsf{tt}}
\newcommand{\accfalse}{\mathsf{ff}}

\newcommand{\wordletter}[2]{#1{[#2]}}

\newcommand{\vertex}[2]{\langle {#1}, {#2} \rangle}

\newcommand{\buchi}{B\"uchi\xspace}

\newcommand{\alphabet}{\Sigma}

\newcommand{\posbool}[1]{\mathcal{B}^+(#1)}

\newcommand{\infwords}{\alphabet^\omega}

\newcommand{\langsymb}[0]{\mathcal{L}}
\newcommand{\lang}[1]{\langsymb(#1)}

\newcommand{\states}{Q}

\newcommand{\trans}{\delta}

\newcommand{\inits}{I}

\newcommand{\run}{\rho}

\newcommand{\naturals}{\mathbb{N}}

\newcommand{\setnocond}[1]{\{#1\}}
\newcommand{\setcond}[2]{\{\, #1 \mid #2 \,\}}

\newcommand{\acc}{F}
\newcommand{\infset}{\textsf{inf}}
\newcommand{\tpo}{\ensuremath{\mathsf{tpo}}}
\newcommand{\nxt}{\ensuremath{\mathsf{next}}}

\begin{document}

\maketitle

\begin{abstract}
We introduce a method for translating an alternating weak \buchi automaton (AWA), which corresponds to a Linear Dynamic Logic (LDL) formula, to an unambiguous \buchi automaton (UBA).
Our translations generalise constructions for Linear Temporal Logic (LTL), a less expressive specification language than LDL.
In classical constructions, LTL formulas are first translated to alternating \emph{very weak} automata (AVAs)---automata that have only singleton strongly connected components (SCCs);
the AVAs are then handled by efficient disambiguation procedures.
However, general AWAs can have larger SCCs, which complicates disambiguation.
Currently, the only available disambiguation procedure has to go through an intermediate construction of nondeterministic B\"uchi automata (NBAs), which would incur an exponential blow-up of its own.
We introduce a translation from \emph{general} AWAs to UBAs with a \emph{singly} exponential blow-up, which also immediately provides a singly exponential translation from LDL to UBAs.
Interestingly, the complexity of our translation is \emph{smaller} than the best known disambiguation algorithm for NBAs (broadly $(0.53n)^n$ vs.\ $(0.76n)^n$),
while the input of our construction can be exponentially more succinct.
\end{abstract}

\section{Introduction}
Automata over infinite words were first introduced by B\"uchi~\cite{Buc62}.
The automata used by B\"uchi (thus called \emph{B\"uchi automata}) accept an infinite word if they have a run over the word that visits accepting states infinitely often.
Nondeterministic \buchi automata (NBAs) are nowadays recognized as a standard tool for model checking transition systems against temporal specification languages like Linear Temporal Logic (LTL)~\cite{DBLP:conf/lics/VardiW86,DBLP:journals/tse/Holzmann97,DBLP:conf/cav/GastinO01,DBLP:books/daglib/0020348}.

NBAs belong to a larger class of automata over infinite words, also known as $\omega$-automata.
Translations between different types of $\omega$-automata play a central role in automata theory, and many of them have gained practical importance, too.
For example, researchers have started to pay attention to a kind of automata called \emph{alternating automata}~\cite{DBLP:conf/litp/MullerS84,DBLP:journals/tcs/MiyanoH84} in the 80s. Alternating automata not only have existential, but also \emph{universal} branching.
In alternating automata, the transition function no longer maps a state and a letter to a set of states, but to a positive Boolean formula over states.
An alternating \buchi automaton accepts an infinite word if there is a run graph over the word, in which all traces visit accepting states infinitely often.
Every NBA can be seen as a special type of alternating \buchi automaton (ABA), while the translation from ABAs to NBAs may incur an exponential blow-up in the number of states~\cite{DBLP:journals/tcs/MiyanoH84}.
Indeed, ABAs can be exponentially more succinct than their counterpart NBAs~\cite{DBLP:conf/icalp/BokerKR10}.
Apart from their succinctness, another reason why alternating automata have become popular in our community is their tight connection to specification logics.
There is a straight forward translation from Linear Dynamic Logic (LDL)~\cite{DBLP:journals/corr/Vardi11,DBLP:conf/ijcai/GiacomoV13} to \emph{alternating weak \buchi automata} (AWAs), both recognizing exactly the $\omega$-regular languages.
AWAs are a special type of ABAs in which every strongly connected component (SCC) contains either only accepting states or only rejecting states.
(AWAs have also been applied to the complementation of \buchi automata~\cite{DBLP:journals/tocl/KupfermanV01}.)
Further, there is a one-to-one mapping~\cite{DBLP:conf/cav/GastinO01,DBLP:conf/ictac/BlahoudekMS19,DBLP:conf/fossacs/BokerLS22} between LTL and \emph{very weak} alternating \buchi automata (AVAs)~\cite{RohdePhD97}---special AWAs where every SCC has only one state.

Automata over infinite words with different branching mechanisms all have their place in building the foundation of the automata-theoretic model checking.
This paper adds another chapter to the success story of efficient automata transformations:
we show how to efficiently translate AWAs to unambiguous B\"uchi automata (UBAs)~\cite{DBLP:journals/tcs/CartonM03}, and thus also the logics that tractably reduce to AWAs, e.g., LDL.
UBAs are a type of NBAs that have at most one accepting run for each word and have found applications in probabilistic verification~\cite{DBLP:conf/cav/BaierK0K0W16}\footnote{We note that specialized model checking algorithm for Markov chains against AWAs/LDL, without constructing UBAs, has been proposed in \cite{DBLP:conf/cav/BustanRV04} without implementations.
Nonetheless, our translation can potentially be used as a third party tool that constructs UBA from an AWA/LDL formula for \textsf{PRISM} model checker~\cite{DBLP:conf/cav/KwiatkowskaNP11} without changing the underlying model checking algorithm~\cite{DBLP:conf/cav/BaierK0K0W16}.}.

Our approach can be viewed as a generalization of earlier work on the disambiguation of AVAs~\cite{DBLP:conf/tacas/BenediktLW13,DBLP:journals/fmsd/JantschMBK21}.
The property of the very weakness has proven useful for disambiguation: to obtain an unambiguous generalized B\"uchi automaton (UGBA) from an AVA, it essentially suffices to use the nondeterministic power of the automaton to guess, in every step, the precise set of states from which the automaton accepts.
There is only one correct guess (which provides unambiguity), and discharging the correctness of these guesses is straight forward.
AVAs with $n$ states can therefore be disambiguated to UGBAs with $2^n$ states and $n$ accepting sets, and thus to UBAs with $n2^n$ states.

Unfortunately, this approach does not extend easily to the disambiguation of AWAs: while there would still be exactly one correct guess, the straight-forward way to discharging its correctness would involve a breakpoint construction~\cite{DBLP:journals/tcs/MiyanoH84}, which is \emph{not} unambiguous.

The technical contribution of this paper is to replace these breakpoint constructions by \emph{total preorders}, and showing that there is a \emph{unique} correct way to choose these orders.
We provide two different reductions, one closer to the underpinning principles---and thus better for a classroom (cf.\ Section~\ref{sec:finite-repre})---and a more efficient approach (cf.\ Section~\ref{sec:improved-construction}).

Given that we track total preorders, the worst-case complexity arises when all, or almost all, states are in the same component.
To be more precise, if $\tpo(n)$ denotes the number of total preorders on sets with $n$ states, then our construction provides UBAs of size $O\big(\tpo(n)\big)$.
As $\tpo(n) \approx \frac{n!}{2(\ln 2)^{n+1}}$ \cite{BARTHELEMY1980311}, we have that $\lim_{n \rightarrow \infty} \frac{\sqrt[n]{\tpo(n)}}{n} = \frac{1}{e\ln 2} \approx 0.53$, which is a better bound than the best known bound for B\"uchi disambiguation \cite{DBLP:conf/atva/KarmarkarJC13} (and complementation \cite{DBLP:conf/stacs/Schewe09}), where the latter number is $\approx 0.76$.

While it is not surprising that a direct construction of UBAs for AWAs is superior to a construction that goes through nondeterminization (and thus incurs two exponential blow-ups on the way), we did not initially expect a construction that leads to a smaller increase in the size when starting from an AWA compared to starting from an NBA, as AWAs can be exponentially more succinct than NBAs, but not vice versa (See \cite{DBLP:journals/tocl/KupfermanV01} for a quadratic transformation).

As a final test for the quality of our construction, we briefly discuss how it behaves on AVAs, for which efficient disambiguation is available.
We show that the complexity of our construction can be improved to $n2^n$ when the input is an AVA, leading to the same construction as the classic disambiguation construction for LTL/AVAs \cite{DBLP:conf/tacas/BenediktLW13,DBLP:journals/fmsd/JantschMBK21} (cf.\ Section~\ref{sec:conclusion}).
We also discuss how to adjust it so that it can produce the same transition based UGBA in this case, too.
The greater generality we obtain comes therefore at no additional cost.
\medskip

\noindent\textbf{Related work. }
Disambiguation of AVAs from LTL specifications have been studied in~\cite{DBLP:conf/tacas/BenediktLW13,DBLP:journals/fmsd/JantschMBK21}.
Our disambiguation algorithm can be seen as a more general form of them.
The disambiguation of NBAs was considered in~\cite{DBLP:conf/icalp/KahlerW08}, which has a blow-up of $\bigO((3n)^n)$;
the complexity has been later improved to $\bigO(n \cdot (0.76n)^n)$ in~\cite{DBLP:conf/atva/KarmarkarJC13}.
Our construction can also be used for disambiguating NBAs, by going through an intermediate construction of AWAs from NBAs;
however, the intermediate procedure itself can incur a quadratic blow-up of states~\cite{DBLP:journals/fmsd/JantschMBK21}.
Nonetheless, if the input is an AWA, our construction improves the current best known approach exponentially by avoiding the alternation removal operation for AWAs~\cite{DBLP:journals/tcs/MiyanoH84,DBLP:conf/icalp/BokerKR10}.

\section{Preliminaries}
For a given set $X$, we denote by $\posbool{X}$ the set of \emph{positive Boolean} formulas over $X$.
These are the formulas obtained from elements of $X$ by only using $\land$ and $\lor$, where we also allow $\acctrue$ and $\accfalse$.
We use $\acctrue$ and $\accfalse$ to represent tautology and contradiction, respectively.
For a set $Y \subseteq X$, we say $Y$ satisfies a formula $\theta \in \posbool{X}$, denoted as $Y \models \theta$, if the Boolean formula $\theta$ is evaluated to $\acctrue$ when we assign $\acctrue$ to members of $Y$ and $\accfalse$ to members of $X\setminus Y$.
For an infinite sequence $\rho$, we denote by $\wordletter{\rho}{i}$ the $i$-th element in $\rho$ for some $i \geq 0$;
for $i \in \naturals$, we denote by $\wordletter{\rho}{i\cdots} = \wordletter{\rho}{i} \wordletter{\rho}{i+1}\cdots$ the suffix of $\rho$ from its $i$-th letter.

An \emph{alternating} \buchi automaton (ABA) $\A$ is a tuple $(\alphabet, \states, \iota, \trans, \acc)$ where $\alphabet$ is a finite alphabet, $\states$ is a finite set of states, $\iota \in \states$ is the initial state, $\trans : \states \times \alphabet \rightarrow \posbool{\states}$ is the transition function, and $\acc \subseteq \states$ is the set of accepting states.
ABAs allow both non-deterministic and universal transitions.
The disjunctions in transition formulas model
the non-deterministic choices, while conjunctions model the universal choices.
The existence of both nondeterministic and universal choices can make ABAs exponentially more succinct than NBAs~\cite{DBLP:conf/icalp/BokerKR10}.
We assume w.l.o.g.\ that every ABA is \emph{complete}, in the sense that there is a next state for each $s \in \states$ and $\sigma \in \alphabet$.
Every ABA can be made complete as follows.
Fix a state $s \in \states$ and a letter $\sigma \in \alphabet$.
If $\trans(s, \sigma) = \accfalse$, we can add a sink rejecting state $\bot$, and set $\trans(s, \sigma) = \bot$ and $\trans(\bot, \sigma) = \bot$ for every $\sigma \in \alphabet$;
If $\trans(s, \sigma) = \acctrue$, we can add a sink accepting state $\top$, and set $\trans(\top, \sigma) = \top$ for every $\sigma \in \alphabet$.
For a state $s \in \states$, we denote by $\A^s$ the ABA obtained from $\A$ by setting the initial state to $s$.

The \emph{underlying graph} $\G_{\A}$ of an ABA $\A$ is a graph $\langle \states, E\rangle$, where the set of vertices is the set $\states$ of states in $\A$ and $(q, q') \in E$ if $q'$ appears in the formula $\trans(q, \sigma)$ for some $\sigma \in \alphabet$.
We call a set $C \subseteq \states$ a \emph{strongly connected component} (SCC) of $\A$ if, for every pair of states $q, q' \in C$, $q$ and $q'$  can reach each other in $\G_{\A}$.

A \emph{nondeterministic \buchi automaton} (NBA) $\A$ is an ABA where $\posbool{\states}$ only contains the $\lor$ operator; we also allow \emph{multiple} initial states for NBAs.
We usually denote the transition function $\trans$ of an NBA $\A$ as a function $\trans: \states \times \alphabet \rightarrow 2^{\states}$ and the set of initial states as $\inits$.
Let $w = \wordletter{w}{0} \wordletter{w}{1} \cdots \in \infwords$ be an (infinite) \emph{word} over $\alphabet$.
A \emph{run} of the NBA $\A$ over $w$ is a state sequence $\rho = q_0 q_1 \cdots \in \states^{\omega}$ such that $q_0 \in \inits$ and, for all $i \in \naturals$, we have that $q_{i+1} \in \trans(q_i, \wordletter{w}{i})$.
We denote by $\infset(\rho)$ the set of states that occur in $\rho$ infinitely often.
A run $\rho$ of the NBA $\A$ is \emph{accepting} if $\infset(\rho) \cap \acc \neq \emptyset$.
An NBA $\A$ accepts a word $w$ if there is an accepting run $\run$ of $\A$ over $w$.
An NBA $\A$ is said to be \emph{unambiguous} (abbreviated as UBA)~\cite{DBLP:journals/tcs/CartonM03} if $\A$ has at most \emph{one} accepting run for every word.

Since ABA have universal branching (or conjunctions in $\trans$), a run of an ABA is no longer an infinite sequence of states;
instead, a run of an ABA $\A$ over $w$ is a run directed acyclic graph (run DAG) $\G_w = (V, E)$ formally defined below:
\begin{itemize}
    \item $V \subseteq \states \times \naturals$ where $\vertex{\iota}{ 0} \in V$.
    \item $E \subseteq \bigcup_{\ell > 0} (\states \times \setnocond{\ell}) \times (\states \times \setnocond{\ell + 1})$ where, for every vertex $\vertex{q}{ \ell} \in V, \ell \geq 0$, we have that $\setcond{q' \in \states}{(\vertex{q}{\ell}, \vertex{q'}{\ell+1}) \in E} \models \trans(q, \wordletter{w}{\ell})$.
\end{itemize}

A vertex $\vertex{q}{\ell}$ is said to be \emph{accepting} if $q \in \acc$.
An infinite sequence $\rho = \vertex{q_0}{0} \vertex{q_1}{1}\cdots$ of vertices is called an \emph{$\omega$-branch} of $\G_w$ if $q_0 = \iota$ and for all $\ell \in \naturals$, we have $(\vertex{q_{\ell}}{\ell}, \vertex{q_{\ell+1}}{\ell + 1}) \in E$.
We also say the fragment $\vertex{q_i}{i}\vertex{q_{i+1}}{i+1}\cdots$ of $\rho$ is an \emph{$\omega$-branch} from $\vertex{q_i}{i}$.
We say a run DAG $\G_w$ is \emph{accepting} if \emph{all} its $\omega$-branches visit accepting vertices infinitely often.
An $\omega$-word $w$ is \emph{accepting} if there is an accepting run DAG of $\A$ over $w$.

Let $\A$ be an ABA.
We denote by $\lang{\A}$ the set of words accepted by $\A$.

It is known that both NBAs and ABAs recognise exactly the $\omega$-regular languages.
ABAs can be transformed into language-equivalent NBAs in exponential time~\cite{DBLP:journals/tcs/MiyanoH84}.
In this work, we consider a special type of ABAs, called \emph{alternating weak B\"uchi automata} (AWAs) where, for every SCC $C$ of an AWA $\A = (\alphabet, \states, \iota, \trans, \acc)$, we have either $C \subseteq \acc$ or $C \cap \acc = \emptyset$.
We note that different choices of equivalent transition formulas, e.g., $\trans(p, \sigma) = q_1$ and $\trans(p, \sigma) = q_1 \land (q_1 \lor q_2)$, will result in different SCCs.
However, as long as the input ABA is weak\footnote{To make ABAs as weak as possible, one solution would be to allow for minimal satisfying assignments to the transition formulas, which is well defined and results in minimal possible SCCs.}, our proposed translation still applies.

One can transform an ABA to its equivalent AWA with only quadratic blow-up of the number of states~\cite{DBLP:journals/tocl/KupfermanV01}.
A nice property of an AWA $\A$ is that we can easily define its dual AWA $\dual{\A} = (\alphabet, \states, \iota, \dual{\trans}, \dual{\acc})$, which has the same statespace and the same underlying graph as $\A$, as follows:
for a state $q \in \states$ and $a \in \alphabet$, $\dual{\trans}(q, a)$ is defined from $\trans(q, a)$ by exchanging the occurrences of $\accfalse$ and $\acctrue$ and the occurrences of $\land$ and $\lor$, and $\dual{\acc} = \states \setminus \acc$.
It follows that:
\begin{lemma}[\hspace*{-1mm}\cite{DBLP:journals/tcs/MullerSS92}]
    Let $\A$ be an AWA and $\dual{\A}$ its dual AWA. For every state $q \in \states$, we have $\lang{\A^{q}} = \infwords \setminus \lang{\dual{\A}^{q}}$.
\end{lemma}

In the remainder of the paper, we call a state of an NBA a \emph{macrostate} and a run of an NBA a \emph{macrorun} in order to distinguish them from those of ABA.

\section{From AWAs to UBAs}

In this section, we will present a construction of UBA $\B_u$ from an AWA $\A$ such that $\lang{\B_u} = \lang{\A}$.
We will first introduce the construction of an NBA from an AWA given in \cite{DBLP:conf/cav/BustanRV04} and show that this construction does \emph{not} necessarily yield a UBA (Section~\ref{ssec:awa-nba}).
Nonetheless, we extract the essence of the construction and show that we can associate a \emph{unique} sequence to each word (Section~\ref{sec:unique-sequence}).

We then enrich this unique sequence with additional, similarly unique, information, which we subsequently abstract into the essence of a unique accepting macrorun of $\B_u$.
Developing this into a UBA whose macrorun can be uniquely mapped to the sequence (Section~\ref{sec:finite-repre}) is then just a simple technical exercise.

\subsection{From AWAs to NBAs}
\label{ssec:awa-nba}

As shown in~\cite{DBLP:journals/tcs/MiyanoH84}, we can obtain an equivalent NBA $\N(\A)$ from an ABA $\A$ with an exponential blow-up of states, which is widely known as the \emph{breakpoint construction}.
In \cite{DBLP:conf/cav/BustanRV04}, the authors define a different construction of NBAs $\B$ from AWAs $\A$, which can be seen as a combination of the NBAs $\N(\A)$ and $\N(\dual{\A})$.
Below we will first introduce the construction in \cite{DBLP:conf/cav/BustanRV04} and extract its essence as a unique sequence of sets of states for each word.

The macrostate of $\B$ is encoded as a \emph{consistent} tuple $(Q_1, Q_2, Q_3, Q_4)$ such that $\states_2 = \states \setminus \states_1, \states_3 \subseteq \states_1 \setminus \acc$, and $\states_4 \subseteq \states_2 \setminus \dual{\acc}$.
The formal translation is defined as follows.
\begin{definition}[\hspace*{-1mm}\cite{DBLP:conf/cav/BustanRV04}]\label{def:cav-construction}
    Let $\A = (\alphabet, \states,\iota, \trans, \acc)$ be an AWA. We define an NBA $\B = (\alphabet, Q_{\B}, \inits_{\B}, \trans_{\B}, \acc_{\B})$ where
    \begin{itemize}
        \item $Q_{\B}$ is the set of consistent tuples over $2^{Q} \times 2^{Q} \times 2^{Q} \times 2^{Q}$.
        \item $I_{\B} = \setcond{(Q_1, Q_2, Q_3, Q_4)\in Q_{\B}}{ \iota \in Q_1}$\footnote{$I_{\B}$ is not present in \cite{DBLP:conf/cav/BustanRV04} and we added it for the completeness of the definition.},
        \item Let $(Q_1, Q_2, Q_3, Q_4)$ be a macrostate in $Q_{\B}$ and $\sigma \in \alphabet$.

        Then $(Q'_1, Q'_2, Q'_3, Q'_4) \in \trans_{\B}((Q_1, Q_2, Q_3, Q_4), \sigma)$ if $Q'_1 \models \land_{s \in Q_1} \trans(s, \sigma)$ and $Q'_2 \models \land_{s \in Q_2} \dual{\trans}(s, \sigma)$ and either
        \begin{itemize}
            \item $Q_3 = Q_4 = \emptyset, Q'_3 = Q'_1\setminus \acc$ and $Q'_4 = Q'_2\setminus\dual{\acc}$,
            \item $Q_3 \neq \emptyset$ or $\states_4 \neq \emptyset$, there exists $Y_3\subseteq Q'_1$ such that $Y_3 \models \land_{s\in Q_3} \trans(s, \sigma)$ and $Q'_3 = Y_3\setminus\acc$, and there exists $Y_4\subseteq Q'_2$ such that $Y_4 \models \land_{s\in Q_4} \dual{\trans}(s, \sigma)$ and $Q'_4 = Y_4\setminus\dual{\acc}$.
        \end{itemize}
        \item $\acc_{\B} = \setcond{(Q_1, Q_2, Q_3, Q_4) \in Q_{\B}}{Q_3 = Q_4 = \emptyset}$.
    \end{itemize}
\end{definition}
Intuitively, the resulting NBA performs two breakpoint constructions:
one breakpoint construction macrostate $(Q_1, Q_3)$ for $\A$ and the other breakpoint construction macrostate $(Q_2, Q_4)$ for $\dual{\A}$.
Let $w \in \alphabet$.
The tuple $(Q_1, Q_3)$ in the construction uses $Q_1$ to keep track of the reachable states of $\A$ in a run DAG $\G_w$ over $w$ and exploits the set $Q_3$ to check whether all $\omega$-branches end in accepting SCCs.
If all $\omega$-branches in $Q_3$ have visited accepting vertices, $Q_3$ will fall empty, as $Q_3$ only contains non-accepting states.
Once $Q_3$ becomes empty, we reset the set with $Q'_3 = Q'_1\setminus\acc$ since we need to also check the $\omega$-branches that newly appear in $Q_1$.
If $Q_3$ becomes empty for infinitely many times, we know that every $\omega$-branch in $\G_w$ is accepting, i.e., all $\omega$-branches visit accepting vertices infinitely often.
Hence $w$ is accepted by $\A$ since there is an accepting run DAG from $\A^{\iota}$.
We can similarly reason about the breakpoint construction for $\dual{\A}$.

Besides that $\lang{\B} = \lang{\A}$, Bustan, Rubin, and Vardi \cite{DBLP:conf/cav/BustanRV04} have also shown the following:
\begin{lemma}[\hspace{-0.1mm}\cite{DBLP:conf/cav/BustanRV04}]\label{lemma:cav-macrostates}
Let $\B$ be the NBA constructed as in Definition~\ref{def:cav-construction}. Then
\begin{itemize}
    \item Let $S \subseteq \states$, we have that
\[  \lang{\B^{(S, \states\setminus S, Q_3, Q_4)}} = \bigcap_{s\in S} \lang{\A^s} \cap \bigcap_{s\in  \states\setminus S}\lang{\dual{\A}^s} \]
where $Q_3 \subseteq S$ and $Q_4\subseteq Q\setminus S$;
    \item Let $(Q_1, Q_2, Q_3, Q_4)$ and $(Q'_1, Q'_2, Q'_3, Q'_4)$ be two macrostates of $\B$, we have that
    \begin{itemize}
        \item $ \lang{\B^{(Q_1, Q_2, Q_3, Q_4)}} \cap \lang{\B^{(Q'_1, Q'_2, Q'_3, Q'_4)}} = \emptyset$ if $Q_1 \neq Q'_1$, and
        \item $ \lang{\B^{(Q_1, Q_2, Q_3, Q_4)}} = \lang{\B^{(Q'_1, Q'_2, Q'_3, Q'_4)}} $ if $Q_1 = Q'_1$.
    \end{itemize}
\end{itemize}
\end{lemma}

Let $w \in \lang{\B}$ and $\rho = (Q^0_1, Q^0_2, Q^0_3, Q^0_4)(Q^1_1, Q^1_2, Q^1_3, Q^1_4)\cdots$ be an accepting macrorun of $\B$ over $w$.
According to Lemma~\ref{lemma:cav-macrostates}, it is easy to see that the $Q_1$-set sequence $Q^0_1 Q^1_1 \cdots$ is in fact \emph{unique} for every accepting macrorun over $w$.
If there are two accepting macroruns, say $\rho_1$ and $\rho_2$, of $\B$ over $w$ that have two different $Q_1$-set sequences, there must be a position $j \geq 0$ such that their $Q_1$-sets differ.
By Lemma~\ref{lemma:cav-macrostates}, the suffix $\wordletter{w}{j\cdots}$ cannot be accepted from both macrostates $\wordletter{\rho_1}{j}$ and $\wordletter{\rho_2}{j}$, leading to contradiction.
Therefore, every accepting macrorun of $\B$ over $w$ corresponds to a unique sequence of $Q_1$-sets.
However, $\B$ does not necessarily have only one accepting macrorun over $w$, because there is \emph{nondeterminism} in developing the breakpoints.

\begin{lemma}\label{lem:cav-not-UBA}
    The NBA $\B$ defined as in Definition~\ref{def:cav-construction} is not necessarily unambiguous.
\end{lemma}
\begin{figure}
    \centering
    \includegraphics{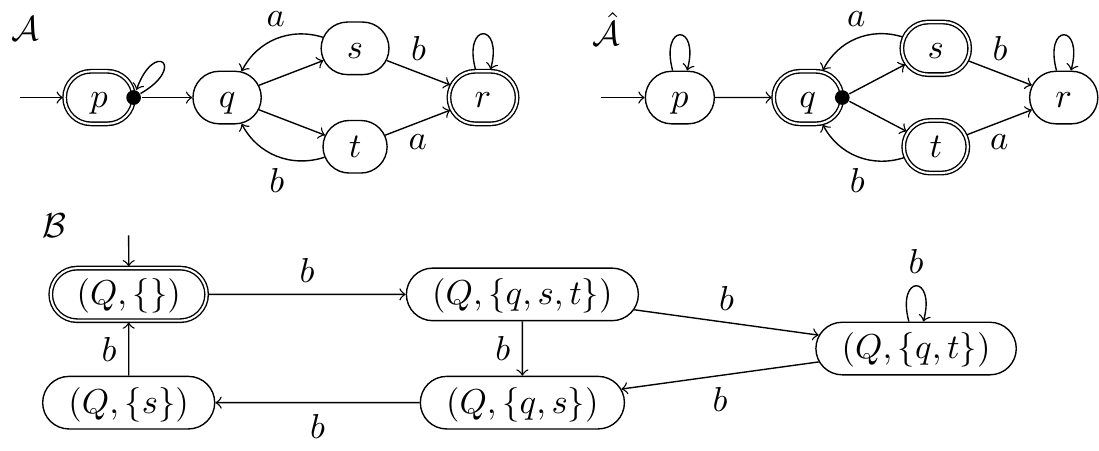}
    \caption{An example of an AWA $\A$, its dual $\dual{\A}$ and \emph{incomplete} part of the constructed $\B$ over $b^{\omega}$, where for instance the transition $((Q, \setnocond{q, s}), b, (Q, \setnocond{t}))$ is missing.}
    \label{fig:awa-dual}
\end{figure}
\begin{proof}
We prove Lemma~\ref{lem:cav-not-UBA} by giving an example AWA $\A$ for which the constructed $\B$ is \emph{not} unambiguous.
    The example AWA $\A$ and its dual $\dual{\A}$ are given in Figure~\ref{fig:awa-dual} where accepting states are depicted with double circles, initial states are marked with an incoming arrow and universal transitions are originated from a black filled circle.
    The transitions are by default labelled with $\alphabet = \setnocond{a,b}$ unless explicitly labelled otherwise.
    We let $Q = \setnocond{p,q,s,t, r}$.
    First, we can see that $b^{\omega} \in \lang{\A^p} \cap \lang{\A^q} \cap \lang{\A^s} \cap \lang{\A^t} \cap\lang{\A^{r}}$.
    So the unique $Q_1$-sequence of all accepting macroruns in $\B$ over $b^{\omega}$ should be $Q^{\omega}$, according to Lemma~\ref{lemma:cav-macrostates}.
    We only depict an \emph{incomplete} part of $\B$ over $b^{\omega}$ where we ignore the $Q_2$ and $Q_4$ sets because we have constantly $Q_2 = \setnocond{}$ and $Q_4 = \setnocond{}$ by definition.
    One of the initial macrostates is $m_0 = (Q,  \setnocond{})$, which is also accepting.
    When reading the letter $b$, we always have $\setnocond{p, q, s, t, r} \models \land_{c \in Q}\trans(c, b)$.
    Thus, the successor of $m_0$ over $b$ is $m_1 = (Q, Q\setminus \setnocond{p,r}) = (Q, \setnocond{q,s,t})$ since the breakpoint set $Q'_3$ needs to be reset to $Q'_1
    \setminus\acc$ when $Q_3=\{\}$.
    When choosing the successor set $Q'_3$ for $Q_3=\setnocond{q,s,t}$ at $m_1$, we have two options, namely $\setnocond{q,s}$ and $\setnocond{q,t}$, since $q$ has nondeterministic choices upon reading letter $b$.
    Consequently, $\B$ can transition to either $m_2 = (Q, \setnocond{q,s})$ or $m_3 = (Q, \setnocond{q,t})$, upon reading $b$ in $m_1$.
    In fact, all the nondeterminism of $\B$ in Figure \ref{fig:awa-dual} are caused by nondeterministic choices at $q$.
    We can continue to explore the state space of $\B$ according to Definition~\ref{def:cav-construction} and obtain the incomplete part of $\B$ depicted in Figure~\ref{fig:awa-dual}.
    Note that, we have ignored some outgoing transitions from $(Q, \setnocond{q, s})$ since the present part already suffices to prove Lemma~\ref{lem:cav-not-UBA}.
    It is easy to see that $\B$ has at least two accepting macroruns over $b^{\omega}$.
    Thus we have proved Lemma~\ref{lem:cav-not-UBA}.
\end{proof}

In fact, based on Definition~\ref{def:cav-construction}, it is easy to compute a unique sequence of sets of states for each given word, which builds the foundation of our proposed construction.

\subsection{Unique sequence of sets of states for each word}
\label{sec:unique-sequence}
In the remainder of the paper, we fix an AWA $\A = (\alphabet, \states, \iota, \trans, \acc)$.
For every word $w \in \infwords$, we define a \emph{unique} sequence of sets of states associated with it as the sequence $Q^0_1 Q^1_1 Q^2_1\cdots$ such that, for every $i\geq 0$, we have that:
\begin{itemize}
    \item[\hspace{6mm}P1]\label{itm:sbst} $Q^i_1 \subseteq \states$,
    \item[\hspace{6mm}P2]\label{itm:gSeqPos} for every state $q \in \states^i_1$, $w[i\cdots] \in \mathcal L(\A^q)$ and
    \item[\hspace{6mm}P3]\label{itm:gSeqNeg} for every state $q\in\states\setminus Q^i_1$, $w[i\cdots] \notin \mathcal L(\A^q)$ \hfill \mbox{(or, similarly, $w[i\cdots] \in \mathcal L(\widehat{\A}^q)$).}
\end{itemize}
These properties immediately entail the weaker \emph{local} consistency requirements:
\begin{itemize}
    \item[\hspace{6mm}L2]\label{itm:SeqPos} for every state $q \in \states^i_1$, $Q^{i+1}_1 \models \trans(q, \wordletter{w}{i})$ \hfill \mbox{(entailed by P2) and}
    \item[\hspace{6mm}L3]\label{itm:SeqNeg} for every state $q\in\states\setminus Q^i_1$, $ \states\setminus Q^{i+1}_1\models \dual{\trans}(q, \wordletter{w}{i})$ \hfill \mbox{(entailed by P3).}
\end{itemize}

It is obvious that,  for every state $s\in\states$, $\infwords = \lang{\A^s} \uplus \overline{\lang{\A^s}} = \lang{\A^s} \uplus \lang{\dual{\A}^s}$ holds.
We define $Q_{w} = \setcond{s \in \states}{ w \in \lang{\A^s}}$.
This clearly provides $\states\setminus Q_{w} = \setcond{s \in \states}{ w \in \lang{\dual{\A}^s}}$.
For every $w \in \infwords$, we therefore have
\[ w \in \bigcap_{s\in Q_w} \lang{\A^s} \cap \bigcap_{s\in  \states\setminus Q_w}\overline{\lang{\A^s}}\text{ or, equivalently, } w \in \bigcap_{s\in Q_w} \lang{\A^s} \cap \bigcap_{s\in  \states\setminus Q_w}\lang{\dual{\A}^s}.\]
For every $i \geq 0$, P2 and P3 are then equivalent to the requirement $Q^i_1 = Q_{\wordletter{w}{i\cdots}}$.

To see how the local constraints L2 and L3 can be obtained from P2 and P3, respectively, we fix an integer $i\geq 0$.
Let $s\in Q^{i}_1$, so we know that $\A^s$ accepts $\wordletter{w}{i\cdots}$.
Let $S^{i+1}$ be the set of successors of $s$ in an accepting run DAG of $\A^s$ over $\wordletter{w}{i\cdots}$, i.e., $S^{i+1}\models \trans(s, \wordletter{w}{i})$.
As the run DAG is accepting, we in particular have, for every $t \in S^{i+1}$, that $\A^{t}$ accepts $\wordletter{w}{i+1 \cdots}$, which implies $S^{i+1} \subseteq Q^{i+1}_1$.
With $S^{i+1}\models \trans(s, \wordletter{w}{i})$, this provides $Q^{i+1}_1 \models \trans(s, \wordletter{w}{i})$, and thus L2.

Similarly, we can also show that, for every state $q \in \states \setminus Q^{i}_1$, we have $\states\setminus Q^{i+1}_1 \models \dual{\trans}(q, \wordletter{w}{i})$.
As before, $\dual{\A}^q$ accepts $\wordletter{w}{i\cdots}$ for every $q \in \states\setminus Q^{i}_1$ by definition.
We let $S^{i+1}$ be the set of successors of $q$ in an accepting run DAG of $\dual{\A}^q$.
This implies at the same time $S^{i+1}\models\dual{\trans}(q, \wordletter{w}{i})$ (local constraints for the run DAG) and $S^{i+1} \subseteq \states\setminus Q^{i+1}_1$ (as the subgraphs starting there must be accepting).
Together, this provides $\states\setminus Q^{i+1}_1 \models \dual{\trans}(q, \wordletter{w}{i})$, and thus L3 also holds.

Moreover, every set $Q^i_1$ is uniquely defined based on the word $\wordletter{w}{i\cdots}$.
Therefore, the sequence $\R_w = Q^0_1Q^1_1\cdots$ we have defined above indeed is the unique sequence satisfying P1, P2, and P3.
Let us consider again the NBA construction of Definition~\ref{def:cav-construction}:
obviously, it enforces the local consistency requirements L2 and L3 on the definition of the transition relation $\trans_{\B}$, which is the necessary condition for the $Q_1$-sequence being unique;
the sufficient condition that $Q^i_1 = Q_{\wordletter{w}{i\cdots}}$ for all $i \in \naturals$ is guaranteed with the two breakpoint constructions.

In the remainder of the paper, we denote this unique sequence for a given word $w$ by $\R_{w}$.
The UBA we will construct has to guess (not only) this unique sequence correctly on the fly, but also when it leaves each SCC, as shown later.

\subsection{Unique distance functions}
\label{ssec:distance}
As discussed before, we have a unique sequence $\R_{w} = Q^0_1Q^1_1\cdots$ for $w$.
However, as we have seen in Section \ref{ssec:awa-nba}, $\R_w$ alone does not suffice to yield an UBA.
The construction from Section \ref{ssec:awa-nba}, for example, validates that all rejecting SCCs can be left using breakpoints, and we have shown how that leaves leeway w.r.t.\ how these breakpoints are met.
In this section, we discuss a different, an unambiguous (but not finite) way to check the correctness of $\R_w$ by making the minimal time it takes from a state, for the given input word, to leave the rejecting SCC of $\A$ or $\widehat{\A}$ on every branch of this run DAG.
For instance, in Figure~\ref{fig:awa-dual}, it is possible to select different successors for state $q$ when reading a $b$, going to either $s$ or $t$.
One of them will lead to leaving this SCC immediately, either $s$ (when reading a $b$) or $t$ (when reading an $a$).
For acceptance, the choice does not matter---so long as the correct choice is eventually made.
On the word $b^\omega$, for example in $\A$, we could go to $t$ the first $20$ times, and to $s$ only in the $21^{st}$ attempt; the answer to the question `how long does it take to leave the SCC starting in $q$ on this run DAG?' would be $42$.

The \emph{shortest} time, however, is well defined. In the example automaton $\A$, it depends on the next letter: if it is $a$, then the distance is $1$ from $t$, $2$ from $q$, and $3$ from $s$, and when it is $b$, then the distance is $1$ from $s$, $2$ from $q$, and $3$ from $t$.

To reason about the minimal number of steps it takes from a state within a rejecting SCC that needs to leave it, we will define a \emph{distance function}.

Formally, we denote by $R$ the set of states in all rejecting SCCs of $\A$ and $A$ the set of states in all accepting SCCs of $\A$.
For a given word $w$ and its unique sequence $\R_w$, we identify the unique distance%
\footnote{Note that, while the distance is unique, the way does not have to be. To see this, we could just expand the alphabet of $\A$ by adding a letter $c$, and by adding $c$ to the transitions from both $s$ and $t$ to $r$. Then there are two equally short (length $2$) ways from $q$ to $r$ whenever the next letter is $c$.}
to leave a rejecting SCCs at each level $i$ in $\G_w$ by defining a distance function $d_i: (Q^i_1\cap R) \uplus (A\setminus Q^i_1) \rightarrow \naturals^{>0}$ for each $i \in \naturals$.

\begin{definition}\label{def:consistent-distances}
    Let $w$ be a word and $\R_w = Q^0_1 Q^1_1
    \cdots$ be its unique sequence of sets of states.
    We say $\Phi_w = (Q^0_1, d_0) (Q^1_1, d_1)\cdots$ is \emph{consistent} if, for every $i \in \naturals$, we have $(Q^i_1, d_i)$ and $(Q^{i+1}_1, d_{i+1})$ satisfy the following rules:
    \begin{enumerate}
    \item[\bf{R1}.]\label{itm:rej-succ} For every state $p \in R \cap Q^i_1$ that belongs to a rejecting SCC $C$ in $\A$,
    \[a:\ (Q^{i+1}_1\setminus C) \cup \{q \in C \cap Q^{i+1}_1 \mid d_{i+1}(q) \leq d_i(p)-1\}\models \delta(p, \wordletter{w}{i}) \text{ and }\]
    \[b:\ \text{for }d_i(p) > 1, (Q^{i+1}_1\setminus C) \cup \{q \in C \cap Q^{i+1}_1\mid d_{i+1}(q) \leq d_i(p)-2\}\not\models \trans(p, \wordletter{w}{i}) \mbox{ hold}.\]

    \item[\bf{R2}.]\label{itm:acc-succ} For every state $p \in A\setminus Q^i_1$ that belongs to an accepting SCC $C$ in $\A$,
    \[a:\ \big(Q\setminus (Q^{i+1}_1\cup C)\big) \cup \{q \in C \setminus Q^{i+1}_1 \mid d_{i+1}(q) \leq d_i(p)-1\}\models \dual{\trans}(p,\wordletter{w}{i})\text{ and}\]
    \[b:\ \text{for }d_i(q)>1, \big(Q\setminus (Q^{i+1}_1\cup C)\big) \cup \{q \in C \setminus Q^{i+1}_1 \mid d_{i+1}(q) \leq d_i(p)-2\}\not\models \dual{\trans}(p,\wordletter{w}{i}) \text{ hold.}\]
\end{enumerate}
\end{definition}

Intuitively, the distance function is to define a \emph{minimal} number of steps to escape from rejecting SCCs over different accepting run DAGs and \emph{maximal} over different branches of one such run DAG.
For instance, when $d_i(p) = 1$, we have that $Q^{i+1}_1 \setminus C \models \trans(p, \wordletter{w}{i})$ if $p \in Q^i_1 \cap R$, otherwise $\states \setminus (Q^{i+1}_1 \cup C) \models \trans(p, \wordletter{w}{i})$ if $p \in A \setminus Q^i_1$.
It means that $p$ can escape from $C$ within one step from an accepting run DAG $\G_{\wordletter{w}{i\cdots}}$ starting from $\vertex{p}{0}$.
\begin{restatable}{lemma}{uniqueDistance}\label{lem:unique-distance}
For each $w \in \infwords$, there is a unique consistent sequence $\Phi_w =(Q^0_1, d_0) (Q^1_1, d_2) \cdots$ where $Q^0_1 Q^1_1 Q^2_1 \cdots$ is $\R_w$ and $d_0 d_1 \cdots$ is the sequence of distance functions.
\end{restatable}

One can easily construct a consistent sequence of distance functions as follows.
Let $C$ be a rejecting SCC of $\A$; the case for a rejecting SCC of $\dual{\A}$ is entirely similar.
Below, we describe how to obtain a sequence of distance values for each state $q \in C \cap Q^i_1$ with $i \geq 0$ in order to form a consistent sequence $\Phi_w$.
For $q \in C \cap Q^i_1$ at the level $i$, we first obtain an accepting run DAG $\G_{\wordletter{w}{i\cdots}}$ over $\wordletter{w}{i\cdots}$ starting from $\vertex{q}{0}$.
One can define the maximal distance, say $K$, over \emph{all} branches from $\vertex{q}{0}$ to escape the rejecting SCC $C$.
Such a maximal distance value must exist and be a finite value, since all branches will eventually get trapped in accepting SCCs.
For all accepting run DAGs $\G'_{\wordletter{w}{i\cdots}}$ over $\wordletter{w}{i\cdots}$ starting from the vertex $\vertex{q}{0}$, there are only finitely many run DAGs of depth $K$ from $\vertex{q}{0}$;
we denote the finite set of such run DAGs of depth $K$ by $P_{q,i}$.
We then denote the maximal distance over one \emph{finite} run DAG $G_{q,i, K} \in P_{q,i}$ by $K_{G_{q,i, K}}$.
(Note that we set the distance to $\infty$ for a finite branch in $G_{q,i, K}$ if it does not visit a state outside $C$.)
We then set $d_i(q) = \min \setnocond{ K_{G_{q,i, K}}: G_{q,i, K} \in P_{q,i}} \leq K$.
One of $G_{q,i, K}$ must provide the \emph{minimal} value, so that $d_i(q)$ is well defined.
This way, we can define the sequence of distance functions $\textbf{d} = d_0 d_1 \cdots$ for the sequence $\R_w$.
We can also prove that the sequence $\R_w \times \textbf{d}$ is consistent by an induction on all the distance values $k > 0$; We refer to Appendix~\ref{app:unique-distance} for the details.

The proof for the uniqueness of $\textbf{d}$ to $\R_w$ can also be obtained by an induction on the distance value $k > 0$; See also Appendix~\ref{app:unique-distance}.
The intuition is that every consistent sequence of distance functions $\textbf{c}$ does not have smaller distance values than $\textbf{d}$ for every state $q \in C\cap Q^i_1$ (see the construction of $\textbf{d}$ above), and if $\textbf{c}$ does have greater distance values for some state, a violation of the consistency constraints in Definition~\ref{def:consistent-distances} will be found, leading to contradiction.

\subsection{Unique total preorders}
\label{sec:finite-repre}
The range of the sequence $\textbf{d}=d_0d_1d_2\ldots$ of distance functions for $\R_w$ is not a priori bounded by any given \emph{finite} number.
Therefore, we may need \emph{infinite} amount of memory to store $\textsf{d}$.
To allow for an abstraction of $\textsf{d}$ that preserves uniqueness and needs only finite memory, we will abstract the values of each function $d_i$ as families of total \emph{preorders}, $\{\preceq^i_C\}_{C\in \mathcal S}$, where $\mathcal S$ denotes the set of SCCs in the graph of $\mathcal A$.
For a given SCC $C \in \mathcal{S}$, a total preorder $\preceq^i_C$ is a relation defined over $H^i \times H^i$, where $H^i = C \cap Q^i_1$ if $C \subseteq R$ or $H^i = C \setminus Q^i_1$ if $C \subseteq A$;
As usual, $\preceq^i_C$ is \emph{reflexive} (i.e., for each $q \in H^i$, $q \preceq^i_C q$) and \emph{transitive} (i.e., for each $q, r, s \in H^i$, $q \preceq^i_C r$ and $r \preceq^i_C s$ implies $q \preceq^i_C s$).
We also have $q \prec^i_C r$ whenever $q \preceq^i_C r$ but $r \not\preceq^i_C q$.
We write $q \backsimeq^i_C r$ if we have $q \preceq^i_C r$ and $r \preceq^i_C q$.
Since $\preceq^i_C$ is total, for every two states $p,q \in H^i$, we have $p \preceq^i_C q$ or $q \preceq^i_C p$.
Note that $\prec^i_C$ is acyclic:
it is impossible for two states $q, p \in H^i$ satisfying $p \prec^i_C q$ and $q \prec^i_C p$.

Formally, we define a consistent sequence of total preorders as below.

\begin{definition}\label{def:consistent-preorders}
Let $w \in \infwords$ and $\R_w = Q^0_1 Q^1_1 \cdots$ be its unique sequence of sets of states.
 We say $\po_w = (Q^0_1, \{\preceq^0_C\}_{C\in \mathcal S}) (Q^1_1, \{\preceq^{1}_C\}_{C\in \mathcal S}) \cdots$ is \emph{consistent} if, for every $i \in \naturals$, we have that $(Q^i_1, \{\preceq^{i}_C\}_{C\in \mathcal S})$ and $(Q^{i+1}_1, \{\preceq^{i+1}_C\}_{C\in \mathcal S})$ satisfy the following rules:
\begin{itemize}
    \item[\bf R1'.] \label{itm:order-rej-succ}  $\forall q,q'\in C \cap Q^{i}_1 \subseteq R$, we have that
    $\ q \prec^i_C q'$ iff there exists $r \in C\cap Q^{i+1}_1$ such that
    \[a:\ \{r' \in  C\cap Q^{i+1}_1 \mid r' \prec^{i+1}_C r\} \cup (Q^{i+1}_1 \setminus C) \models \trans(q,\wordletter{w}{i}) \text{ and}\]
       \[b:\ \{r' \in C\cap Q^{i+1}_1 \mid r' \prec^{i+1}_C r\} \cup (Q^{i+1}_1 \setminus C) \not\models \trans(q',\wordletter{w}{i}) \text{ hold, } \] where $C \subseteq R$  is a rejecting SCC of $\A$.

    \item[\bf R2'.] \label{itm:order-acc-succ}  $\forall q,q'\in C\setminus Q^{i}_1 \subseteq A$, we have $q \prec^i_C q'$ iff there exists $r \in C\setminus Q^{i+1}_1$ such that
    \[a:\ \{r' \in  C\setminus Q^{i+1}_1 \mid r' \prec^{i+1}_C r\} \cup \big(Q\setminus (Q^{i+1}_1\cup C)\big) \models \dual{\trans}(q,\wordletter{w}{i}) \text{ and }\]
       \[b:\ \{r' \in C\setminus Q^{i+1}_1 \mid r' \prec^{i+1}_C r\} \cup \big(Q\setminus (Q^{i+1}_1\cup C)\big)\not\models \dual{\trans}(q',\wordletter{w}{i}) \text{ hold, }\] where $C \subseteq A$ is an accepting SCC of $\A$.
\end{itemize}
\end{definition}

As the names indicate, the Rules R1' and R2' correspond to Rules R1 and R2, respectively, from Definition~\ref{def:consistent-distances}.
We will first show that there is a consistent sequence of total preorders for each word.
\begin{lemma}\label{lem:exist-preorders}
    For each word $w\in\infwords$, there exists a consistent sequence $\po_w = (Q^0_1,\{\preceq^{0}_C\}_{C\in \mathcal S}) (Q^1_1, \{\preceq^{1}_C\}_{C\in \mathcal S}) \cdots$, where $Q^0_1 Q^1_1 \cdots$ is the unique sequence $\R_w$.
\end{lemma}

\begin{proof}
It is simple to derive a consistent sequence $\po_w = (Q^0_1, \{\preceq^{0}_C\}_{C\in \mathcal S}) (Q^1_1,\{\preceq^{1}_C\}_{C\in \mathcal S})\cdots$ from $\Phi_w = (Q^0_1, d_0) (Q^1_1, d_1) \cdots$ as given in Lemma~\ref{lem:unique-distance}:
We can simply select, for all $i \in \naturals$ and $C\in \mathcal S$, $\preceq^i_C$ is the total preorder over $C\cap Q^{i}_1$ (if $C\subseteq R$) or $C\setminus Q^i_1$ (if $C \subseteq A$) with $p\preceq^i_C q$ iff $d_i(p) \leq d_i(q)$.
In particular, $p \prec^i_C q$ iff $d_i(p) < d_i(q)$.

It is easy to verify that the sequence $\po_w$ as defined above is indeed consistent.
For instance, for all $q, q' \in C\cap Q^i_1 \subseteq R$, if $q \prec^i_C q'$, then $d_i(q) < d_i(q')$ by definition.
Then we can choose the $r$-state in Definition~\ref{def:consistent-preorders} (Rule R1') such that $d_{i+1}(r) = d_i(q') - 1$. (Note that some such a state $r$ must exist since $d_i(q') > d_i(q) \geq 1$.)

Combining Definition~\ref{def:consistent-distances} (R1) and Definition~\ref{def:consistent-preorders} (R1'), we have that
Rule R1b now entails R1'b, and Rule R1a entails R1'b, because $\{r' \in C\cap Q^{i+1}_1 \mid r' \prec_C^{i+1} r\} \supseteq \{r' \in C\cap Q^{i+1}_1 \mid d_{i+1}(r') \leq d_i(q)-1\}$, because $d_i(q)-1 \leq d_i(q')-2 < d_i(q') - 1 = d_{i+1}(r)$.

The argument for accepting SCCs is using rules R2 and R2' in the same way.
\end{proof}

After discussing how such a sequence can be obtained, we now establish that it is unique.
Note, however, that it is unique for the correct sequence $\R_w$, while there may be sequences of total preorders that work with incorrect sequences of sets of states.
(For example, a total preorder can accommodate an infinite distance for all states, where the obligation to leave a rejecting SCC cannot be discharged, while the local consistency constraints can be met.)
Nonetheless, a breakpoint construction ensures to obtain the unique sequence $\R_w$.
\begin{restatable}{lemma}{mappingToPreorders}\label{lem:one-one-mapping}
Let $w$ be a word in $\infwords$ and $\Phi_w =(Q^0_1, d_0)(Q^1_1, d_1) \cdots$ be its unique consistent sequence of distance functions.
Let $\po_w = (Q^0_1, \{\preceq^{0}_C\}_{C\in \mathcal S})(Q^1_1,\{\preceq^{1}_C\}_{C\in \mathcal S})\cdots$ be a sequence satisfying Definition~\ref{def:consistent-preorders}.
Then
\begin{itemize}
    \item For every two states $q, q' \in C \cap Q^i_1 \subseteq R$, if $q \preceq^i_C q'$, then $d_i(q)\leq d_i(q')$, and in particular if $q \prec^i_C q'$, then $d_i(q) < d_i(q')$.
    \hfill \mbox{($C$ is a rejecting SCC)}

    \item For every two states $q, q' \in C \setminus Q^i_1 \subseteq A$, if $q \preceq^i_C q'$, then $d_i(q)\leq d_i(q')$, and in particular if $q \prec^i_C q'$, then $d_i(q) < d_i(q')$.
    \hfill \mbox{($C$ is an accpting SCC)}

\end{itemize}
\end{restatable}
\begin{proof}
    We only prove the first claim; the proof of the second claim is entirely similar.

Let $C$ be a rejecting SCC with $q,q' \in C\cap Q^i_1$.
We prove the claim by an induction over the distance value $k \in \naturals$ for $d_i(q) \leq k$.
In order to prove that $q \preceq^i_C q'$ implies $d_i(q) \leq d_i(q')$, we can just prove its contraposition that $d_i(q') < d_i(q)$ implies $q' \prec^i_C q$ for all distance values $k \in\naturals$ with $d_i(q') \leq k$.
We prove $q\prec^i_C q'$ implies $d_i(q) < d_i(q')$ similarly.

For the \textbf{induction basis} ($k=1$), we have $d_i(q') \leq k$.
So $d_i(q') = 1$.
But then $Q^{i+1}_1\setminus C \models \trans(q',w[i])$.
Consequently, by Rule R1'b, $q'$ must be a minimal element of $\preceq^i_C$, and we have $q' \preceq^i_C q$.
Since by assumption that $d_i(q)> d_i(q') = 1$, Rule R1 supplies $Q^{i+1}_1\setminus C \not\models \trans(q,w[i])$.
We can therefore choose $r$ from Rule R1' as a minimal element of $\preceq^{i+1}_C$ to get $S^{i+1} = \setcond{r' \in C \cap Q^{i+1}_1}{r' \prec^{i+1}_C r} = \emptyset$.
It follows that $S^{i+1} \cup (Q^{i+1}_1\setminus C) \models \trans(q', \wordletter{w}{i})$ (R1'a) but $S^{i+1} \cup (Q^{i+1}_1\setminus C) \not\models \trans(q, \wordletter{w}{i})$ (R1'b).
By Definition~\ref{def:consistent-preorders}, we have $q' \prec^i_C q$.
Hence, for $k \in \naturals$ with $d_i(q') \leq k = 1$, it holds that $d_i(q') < d_i(q)$ implies $q' \prec^i_C q$.
When $d_i(q) = d_i(q') = k = 1$, it directly follows that $q \not\prec^i_C q'$ and $q' \not\prec^i_C q$ by Definition~\ref{def:consistent-preorders}, thus also $q' \simeq^i_C q$ since $\preceq^i_C$ is a total preorder.
Therefore, if $d_i(q') \leq d_i(q)$, then $q' \preceq^i_C q$, thus also $q\prec^i_C q'$ implies $d_i(q) < d_i(q')$.

For the \textbf{induction step} $k \mapsto k+1$, we have $d_i(q')=k+1$ and we want to prove $q' \prec^i_C q$ when $k + 1 = d_i(q') < d_i(q)$, and prove $q' \simeq^i_C q$ when $d_i(q') = d_i(q)$.
First, there must be a state $s \in C \cap Q^{i+1}_1$ with $d_{i+1}(s) = k$ according to R1a in Definition~\ref{def:consistent-distances}; We then pick such a state $s$.
By induction hypothesis, for all $p,p' \in C\cap Q^i_1 \wedge d_i(p) \leq k \wedge d_i(p')> d_i(p) $, we have that $ p \prec^i_C p'$.
Moreover, the following holds for every $s' \in C \cap Q^{i+1}_1$:
(1) $s' \preceq^{i+1}_C s$ iff $d_{i+1}(s') \leq k = d_{i+1}(s)$, and
    (2) $s' \prec^{i+1}_C s$ iff $d_{i+1}(s') < k = d_{i+1}(s)$.

Item (1) implies that $\setcond{s' \in C\cap Q^{i+1}_1 }{ s' \preceq^{i+1} s} = \setcond{s' \in C \cap Q^{i+1}_1}{ d_{i+1}(s') \leq d_{i+1}(s) = k}$, while Item (2) gives that $\setcond{s' \in C\cap Q^{i+1}_1 }{ s' \prec^{i+1} s} = \setcond{s' \in C \cap Q^{i+1}_1}{ d_{i+1}(s') \leq  d_{i+1}(s) - 1= k - 1}$.
Hence, by Definitions~\ref{def:consistent-distances} and~\ref{def:consistent-preorders}, we have that, for state $p \in C\cap Q^i_1$, $d_i(p) \leq k$ iff  $(Q^{i+1}_1\setminus C)\cup \setcond{s' \in C\cap Q^{i+1}_1}{ d_{i+1}(s') \leq d_i(p) - 1 \leq k-1} \models \trans(p, \wordletter{w}{i})$ (by R1 of Definition~\ref{def:consistent-distances}) iff $(Q^{i+1}_1\setminus C) \cup \{s' \in C \cap Q^{i+1}_1 \mid s' \prec^{i+1}_C s\} \models \trans(p,w[i])$ (by Item (2)).

Further, by applying R1 in Definition~\ref{def:consistent-distances},
it follows that, for state $p \in C\cap Q^i_1$,  $d_i(p) = k+1$ iff $(Q^{i+1}_1\setminus C)\cup \setcond{s' \in C\cap Q^{i+1}_1}{ d_{i+1}(s') \leq d_i(p) - 1 = k = d_{i+1}(s) = d_i(p) - 1} \models \trans(p, \wordletter{w}{i})$ (by R1a of Definition~\ref{def:consistent-distances}) iff $(Q^{i+1}_1\setminus C) \cup \{s' \in C\cap Q^{i+1}_1 \mid s' \preceq^{i+1}_C s\} \models \trans(p,w[i])$ (by Item (1), and we denote it as E1).
With this, Rule R1' implies that $p \in C \cap Q^i_1$ with $d_i(p) = k + 1$ must be a minimal element w.r.t.\ $\preceq^i_C$ in the set $\{p \in C \cap Q^i_1 \mid d_i(p) > k\}$;
Otherwise, there must be some $p \in C \cap Q^{i}_1$ with $p \prec^{i}_C q'$ and $d_i(p) \geq k +1 = d_i(q') > k$, and an element $r \in C\cap Q^{i+1}_1 $ satisfying R1'a and R1'b for $p$ and $q'$, violating the induction hypothesis and Definition~\ref{def:consistent-distances}.

Since we have $q' \in C \cap Q^i_1 \land d_i(q') = k + 1$ by assumption, $q'$ is also a minimal element w.r.t.\ $\preceq^i_C$ in the set $\{p \in C \cap Q^i_1 \mid d_i(p) > k\}$.
Obviously, $q' \preceq^i_C q$ holds because $d_i(q') < d_i(q)$.
Moreover, by assumption that $d_i(q)> d_i(q') = k+1$, then
we pick a state $r$ that is minimal w.r.t.\ $\preceq^{i+1}_C$ in the set $\{p \in C \cap Q^{i+1}_1 \mid d_i(p) > k\} $.
Recall that by induction hypothesis, for all $d_i(q') \leq k$, we have that $q' \prec^{i}_C q$ iff $d_{i}(q') < d_i(q) $ for all $i \in \naturals$.
Hence, $r \in \{p \in C \cap Q^{i+1}_1 \mid d_i(p) > k\} = \{p \in C \cap Q^{i+1}_1 \mid s \prec^{i+1}_C p\}$.
Together with $\{p \in C \cap Q^{i+1}_1 \mid p \preceq^{i+1}_C s\} \subseteq \{p \in C \cap Q^{i+1}_1 \mid p \prec^{i+1}_C r\}$ and (E1), $Q^{i+1}\setminus C \cup  \{p \in C \cap Q^{i+1}_1 \mid p \prec^{i+1}_C r\} \models \trans(q', \wordletter{w}{i})$ (R1'a).
Moreover, $Q^{i+1}\setminus C \cup  \{p \in C \cap Q^{i+1}_1 \mid p \prec^{i+1}_C r\} \not\models \trans(q, \wordletter{w}{i})$ ((R1'b) since by induction hypothesis, it is equivalent to $Q^{i+1}\setminus C \cup  \{p \in C \cap Q^{i+1}_1 \mid d_{i+1}(p) \leq d_i(r) - 1 = k \leq d_i(q) - 2 \} \not\models \trans(q, \wordletter{w}{i})$ (see Definition~\ref{def:consistent-distances}).
Then R1' implies $q' \prec^i_C q$.
Hence, we also have that $d_i(q') < d_i(q)$ implies that $q' \prec^i_C q$.

To prove that $q \prec^i_C q'$ implies $d_i(q) < d_i(q')$, we also prove its contraposition, i.e., $d_i(q') \leq d_i(q)$ implies  $q' \preceq^i_C q$ for all $i \in \naturals$.
We have already shown that $d_i(q') < d_i(q)$ implies  $q' \prec^i_C q$.
Moreover, if $d_i(q') = d_i(q) = k + 1$, then $q' \simeq^i_C q$, since both $q'$ and $q$ are minimal element w.r.t. $\preceq^i_C$ in the set $\{p \in C \cap Q^i_1 \mid d_i(p) > k\}$.
It then follows that $q \prec^i_C q'$ implies $d_i(q) < d_i(q')$.
Hence, we have completed the proof.
\end{proof}

By Lemma~\ref{lem:one-one-mapping}, for states $p, q \in H^i$, we have both $p \simeq^i_C q \Longleftrightarrow d_i(p) = d_i(q)$ and $p \prec^i_C q \Longleftrightarrow d_i(p) < d_i(q)$ hold for all $i \in \naturals$, where $H^i = C \cap Q^i_1$ if $C\subseteq R$ and $H^i = C \setminus Q^i_1$ if $C\subseteq A$.
Then Corollary~\ref{cor:unique-preorder} follows immediately from Lemma~\ref{lem:unique-distance}.
\begin{corollary}\label{cor:unique-preorder}
For each $w \in \infwords$, there is a unique consistent sequence of sets of states and total preorders $\po_w =  (Q^0_1, \{\preceq^{0}_C\}_{C\in \mathcal S})(Q^1_1,\{\preceq^{1}_C\}_{C\in \mathcal S})\cdots$ where $Q^0_1 Q^1_1 Q^2_1 \cdots$ is the unique sequence $\R_w$.
\end{corollary}

In order to lift this unique set to an UBA, we need to discharge the correctness of the sequence $Q^0_1 Q^1_1Q^2_1\cdots$.
This is, however, a relatively simple task: for the correct sequence, the total preorders provide the same rational way of creating the same accepting runs on the tails $w[i\cdots]$ of $w$ from the states marked as accepting in $\A$ by inclusion in $Q^i_1$, or as accepting from $\widehat{\A}$ by non-inclusion in $Q^i_1$.

To prepare such a construction, we first define an arbitrary (but fixed) order on the SCCs of $\A$, as well as a $\nxt$ operator for cycling through SCCs, and fix an initial SCC $C_0 \in \mathcal S$.
Recall that $\mathcal{S}$ is the set of all SCCs in $\A$.
Note that we assume that the graph of $\mathcal A$ has at least one SCC. If this is not the case, we can simply build an unambiguous safety automaton that guesses $\R_w$.
Then, our construction of UBA is formalized below.

\begin{definition}\label{def:separated-ba}
        Let $\A = (\alphabet, \states,\iota, \trans, \acc)$ be an AWA. We define an NBA $\B_u = (\alphabet, Q_u, I_u, \trans_u, \acc_u)$ as follows.
    \begin{itemize}
        \item The states of $Q_u$ are tuples
        $(Q_1,Q_2,\{\preceq_C\}_{C \in \mathcal S},S,D)$
        such that
        \begin{itemize}
            \item $Q_1$ and $Q_2$ partition $Q$, i.e., $Q_2 = \states \setminus Q_1$
            \item for all $C\in \mathcal S$, if $C \subseteq R$ then $\preceq_C$ is a total preorder over $Q_1 \cap C$
            \item for all $C\in \mathcal S$, if $C \subseteq A$ then $\preceq_C$ is a total preorder over $Q_2 \cap C$
            \item $S \in \mathcal S$ is an SCC in the graph of $\mathcal A$
            \item $D$ is a downwards closed set w.r.t.\ the total preorder $\preceq_{S}$: if $q \in D$ then (1) $q \in Q_1 \cap S$ if $S \subseteq R$ resp.\ $q \in Q_2 \cap S$ if $S \subseteq A$, and
            (2) $q' \preceq_{S} q$ implies $q' \in D$,
        \end{itemize}

        \item $I_u = \setcond{(Q_1,Q_2,\{\preceq_C\}_{C \in \mathcal S},S,D)\in Q_u}{\iota \in Q_1, S=C_0, D=\emptyset}$,

        \item Let  $(Q_1,Q_2,\{\preceq_C\}_{C \in \mathcal S},S,D)$ be a macrostate in $Q_u$ and $\sigma \in \alphabet$.
        Then we have that \linebreak $(Q_1',Q_2',\{\preceq_C'\}_{C \in \mathcal S},S',D') \in \trans_u\big( (Q_1,Q_2,\{\preceq_C\}_{C \in \mathcal S},S,D), \sigma)$ if
        \begin{itemize}
            \item $Q'_1 \models \land_{s \in Q_1} \trans(s, \sigma)$ and $Q'_2 \models \land_{s \in Q_2} \dual{\trans}(s, \sigma)$
            \hfill \mbox{(local consistency)}
        \item for all $C \in \mathcal S$, $(Q_1,\preceq_C)$ and $(Q_1',\preceq_C')$
        satisfy the requirements of Rule R1' (if $C \subseteq R$) resp.\ Rule R2' (if $C \subseteq A$)
        \item if $D = \emptyset$, then
        $S' = \nxt(S)$ and $D' = Q_1' \cap S'$ if $S' \subseteq R$ resp.\ $D' = Q_2' \cap S'$ if $S' \subseteq A$,
        \item if $D \neq \emptyset$, then $S' = S$ and $D'$ is the smallest downwards closed set (see above) such that
        $D' \cup (Q_1'\setminus S) \models \land_{s \in D} \trans(s, \sigma)$ if $S \subseteq R$ resp.\  $D' \cup (Q_2'\setminus S) \models \land_{s \in D} \dual{\trans}(s, \sigma)$ if $S \subseteq A$,
    \end{itemize}
        \item $\acc_u = \setcond{(Q_1,Q_2,\{\preceq_C\}_{C \in \mathcal S},S,D)\in Q_u}{D=\emptyset}$.
    \end{itemize}
\end{definition}

The new construction uses $D$ as the breakpoint to ensure that the correct unique sequence $\R_w$ for each word $w$ is obtained.
The nondeterminism of the construction lies only in choosing $Q'_1$ (which entails $Q'_2$) and in updating the total preorders.
From an accepting macrorun of $\B_u$ over a word $w$, one can actually construct an accepting run DAG $\G_w$ of $\A$ by selecting successors in the next level for each state $q$ only the ones in the smallest downwards closed set $D$ satisfying $\trans(q, \sigma)$.
This way, all branches of $\G_w$ by construction will eventually get trapped in an accepting SCC, since $D$ will become empty infinitely often.
Hence, $\lang{\B_u} \subseteq \lang{\A}$.
Moreover, one can construct from the unique sequence of preorders $\Phi_w$ of a word $w \in \lang{\A}$ as given in Corollary~\ref{cor:unique-preorder} a unique infinite macrorun $\rho$ of $\B_u$.
Wrong guesses of the preorders for $\R_w$ will result in discontinued macroruns once a violation to R1' (or R2') has been detected.
Further, by Lemma~\ref{lem:one-one-mapping}, we have that $d_i(q) = d_i(q') \Leftrightarrow q \simeq^i_C q'$ and $d_i(q) < d_i(q') \Leftrightarrow q \prec^i_C q'$ for all $i \in \naturals$.
So, by Definition~\ref{def:consistent-distances} and Definition~\ref{def:consistent-preorders}, one can observe that, if $D^i \neq \emptyset$, $\sup\{d_i(q) \mid q \in D^i\} = \sup\{d_{i+1}(q) \mid q \in D^{i+1}\} +1$ (choosing $\sup \emptyset = 0$), where $D^i$ is the $D$-component of macrostate $\wordletter{\rho}{i}$ with $i \in \naturals$.
Since for every nonempty $D^i$, $\sup\{d_i(q) \mid q \in D^i\}$ is finite and the maximal value in $D^i$ is always decreasing, the value will eventually become $0$, i.e., $D$ always becomes empty eventually.
That is, $\rho$ must be accepting.
Hence, Theorem~\ref{thm:correctness-new-construction} follows;
See Appendix~\ref{app:new-construction-correctness} for more details.

\begin{restatable}{theorem}{newConstruction}\label{thm:correctness-new-construction}
\label{theo:Bu}
Let $\B_u$ be defined as in Definition \ref{def:separated-ba}.
Then (1) $\lang{\B_u} = \lang{\A}$,
(2) $\B_u$ is unambiguous.
\end{restatable}

\section{Improvements and Complexity}
\label{sec:improved-construction}
When revisiting the construction in search for improvements, it seems wasteful to keep total preorders for all SCCs in the graph of $\A$, given that they are not interacting with each other.
Can we focus on just one at a time?
It proves to be possible with Definition~\ref{def:separated-ba}.

\begin{definition}\label{def:separated-ba}
        Let $\A = (\alphabet, \states,\iota, \trans, \acc)$ be an AWA. We define an NBA $\U = (\alphabet, Q_u, I_u, \trans_u, \acc_u)$ as follows.
    \begin{itemize}
        \item The states of $Q_u$ are tuples
        $(Q_1,Q_2,\preceq_C,C,D)$
        such that
        \begin{itemize}
            \item $Q_1$ and $Q_2$ partition $Q$
            \item $C$ is an SCC in the graph of $\mathcal A$ and
            \begin{itemize}
                \item if $C \subseteq R$ then $\preceq_C$ is a total preorder of $Q_1 \cap C$
                \item if $C \subseteq A$ then $\preceq_C$ is a total preorder of $Q_2 \cap C$
            \end{itemize}
            \item let $M$ be the set of maximal elements of the total preorder $\preceq_C$, and let $D'=C \cap Q_1$ if $C \subseteq R$ resp.\ $D'=C \cap Q_2$ if $C \subseteq A$; then
            $D = D'$ or $D = D' \setminus M$
        \end{itemize}

        \item $I_u = \setcond{(Q_1,Q_2,\preceq_C,C,D)\in Q_u}{\iota \in Q_1, C=C_0, D=\emptyset}$,

        \item Let  $(Q_1,Q_2,\preceq_C,C,D)$ be a macrostate in $Q_u$ and $\sigma \in \alphabet$.
        Then we have that \linebreak $(Q_1',Q_2',\preceq_{C'}',C',D') \in \trans_u\big( (Q_1,Q_2,\preceq_C,C,D), \sigma)$ if
        \begin{itemize}
            \item $Q'_1 \models \land_{s \in Q_1} \trans(s, \sigma)$ and $Q'_2 \models \land_{s \in Q_2} \dual{\trans}(s, \sigma)$
            \hfill \mbox{(local consistency)}
        \item if $D = \emptyset$, then
        $C' = \nxt(C)$ and $D' = Q_1' \cap C'$ if $C' \subseteq R$ resp.\ $D' = Q_2' \cap C'$ if $C' \subseteq A$,
        \item if $D \neq \emptyset$ then $C' = C$,
        \begin{itemize}

            \item for all $C \in \mathcal S$, $(Q_1,\preceq_C)$ and $(Q_1',\preceq_C')$
            satisfy the requirements of Rule R1' (if $C \subseteq R$) resp.\ Rule R2' (if $C \subseteq A$) and

            \item $D'$ is the smallest downward closed set w.r.t.\ $\preceq_{C}'$ such that%
\footnote{Note that this is a deterministic assignment that does not necessarily lead to a set $D'$ that covers all of $\preceq_{C}'$ or all of $\preceq_{C}'$ except for the maximal elements; if it does not, then this transition is disallowed}
        $D' \cup (Q_1'\setminus C) \models \land_{s \in D} \trans(s, \sigma)$ if $C \subseteq R$ resp.\  $D' \cup (Q_2'\setminus C) \models \land_{s \in D} \dual{\trans}(s, \sigma)$ if $C \subseteq A$,

        \end{itemize}
    \end{itemize}
        \item $\acc_u = \setcond{(Q_1,Q_2,\preceq_C,C,D)\in Q_u}{D=\emptyset}$.
    \end{itemize}
\end{definition}
The nondeterminism of the construction again lies in choosing $Q'_1$ (which entails $Q'_2$) and in updating the total preorder.
One can also construct from an accepting macrorun of $\U$ over $w$ an accepting run DAG $\G_w$ of $\A$, using the same way as we did for Theorem~\ref{thm:correctness-new-construction}.
So, $\lang{\U} \subseteq \lang{\A}$.
For the other direction, we first observe that the preorders of \emph{every} accepting macrorun $(Q_1^0,Q_2^0,\preceq_0,S^0,D^0)
(Q_1^1,Q_2^1,\preceq_1,S^1,D^1)\cdots$ of $\U$ over $w$ can be tightly related with the distance values of states defined in $\textbf{d}$.
More precisely, let $D^{i'} = D^{i} =\emptyset$ with $i' < i$ being two consecutive accepting positions.
Then for all $j \in (i', i]$, we have that:
\begin{enumerate}
    \item for all $q \in D^j$ and all $ q' \in C^i \cap Q_1^j.\ d_j(q)\leq d_j(q') \Leftrightarrow q \preceq_j q'$, and $d_j(q)\leq i-j$ hold,
    \item for all $q \in C^i \cap Q_1^j$ and all $ q' \in M^j=(C^i \cap Q_1^j) \setminus D^j.\  q \preceq_j q'$ and $d_j(q') > i-j$ hold, and
    \item $m_j = \sup\{d_j(q) \mid q \in D^j\} = i-j$, using $\sup\emptyset = 0$,
\end{enumerate}

\noindent where $C^i \subseteq R$ is a rejecting SCC of $\A$.
Note that $C^j = C^i$ for all $i' < j  \leq i$.
The case for $C^i\subseteq A$ can be defined similarly.
Let $m_j = \sup\{d_j(q) \mid q \in D^j\}$.
The intuition is that all states in $M^j = (C^i \cap Q^j_1)\setminus D^j = \setcond{s \in C^i \cap Q^j_1}{d_j(s) > m_i}$ are aggregated by construction as the maximal elements w.r.t. $\preceq_j$, while $\preceq_j$ orders all states in $D^j = \setcond{s \in C^i\cap Q^j_1}{d_j(s) \leq m_j}$ exactly as in the preorders of Corollary~\ref{cor:unique-preorder}.
So, the correspondence between $d_j$ and $\preceq_j$ in the three items then follows naturally.
For technical reasons, if $q\in D^j$ or $q' \in (C^i \cap Q^j_1)\setminus D^j$ do not exist in above items, we say the item above still holds.
See Appendix~\ref{app:improved-construction} for proof details.

In fact, one can construct such an accepting macrorun satisfying the three items above for $\U$ by simulating $\B_u$ as follows.
If
$\rho = (Q_1^0,Q_2^0,\{\preceq_C^0\}_{C \in \mathcal S},S^0,D^0)
(Q_1^1,Q_2^1,\{\preceq_C^1\}_{C \in \mathcal S},S^1,$ $D^1)
(Q_1^2,Q_2^2,\{\preceq_C^2\}_{C \in \mathcal S},S^2,D^2)\cdots$ is the accepting macrorun of $\B_u$ on a word $w$, then $\U$ has an accepting macrorun
$\dual{\rho} = (Q_1^0,Q_2^0,\preceq_0,S^0,D^0)
(Q_1^1,Q_2^1,\preceq_1,S^1,D^1)
(Q_1^2,Q_2^2,\preceq_2,S^2,D^2)\cdots$ (that differs from $\rho$ only in preorders), such that
\begin{itemize}
    \item if $S^i \subseteq R$, then $\preceq_i$ is a total preorder on $S^i \cap Q_1^i$ where
        $\preceq_i = \preceq_{S^i}^i$ if $D^i=S^i \cap Q_1^i$ and
        otherwise, the maximal elements $M^i$ of $\preceq_i$ are $(S^i \cap Q_1^i) \setminus D^i$, and the restriction of $\preceq_i$ to $D^i \times D^i$ agrees with the restriction of $\preceq_{S^i}^i$ to $D^i \times D^i$, and
    \item similarly, if $S^i \subseteq A$, then $\preceq_i$ is a total preorder on $S^i \cap Q_2^i$ where
        $\preceq_i = \preceq_{S^i}^i$ if $D^i=S^i \cap Q_2^i$ and
        otherwise, the maximal elements $M^i$ of $\preceq_i$ are $(S^i \cap Q_2^i) \setminus D^i$, and the restriction of $\preceq_i$ to $D^i \times D^i$ agrees with the restriction of $\preceq_{S^i}^i$ to $D^i \times D^i$.
\end{itemize}

It is easy to verify that $\dual{\rho}$ satisfies all local constraints for Rule R1' resp.\ R2'.
Hence, $\lang{\A} = \lang{\B_u} \subseteq \lang{\U}$, thus also $\lang{\U} = \lang{\A}$.
One can show that $\dual{\rho}$ is the sole accepting macrorun of $\U$ over $w$ by the following facts.
(i) There is only a single initial macrostate that fits $\R_w$, and when we take a transition from an accepting macrostate (including the first), the next SCC is deterministically selected;
(ii) Moreover, all relevant states from this SCC are in the $D^i$ component and $m_i = \sup\{d_i(q) \mid q \in D^i\}$ is the distance to the next breakpoint (by Item (3) above), and thus the $\preceq_i$ and $D^i$ up to it.
With a simple inductive argument we can thus conclude that $\dual{\rho}$ is the only such accepting macrorun.
Then, Theorem~\ref{thm:correctness-improved-construction} follows.
\begin{theorem}\label{thm:correctness-improved-construction}
Let $\U$ be defined as in Definition \ref{def:separated-ba}.
Then (1) $\lang{\U} = \lang{\A}$
and
(2) $\U$ is unambiguous.
\end{theorem}

We now turn to the complexity of our constructions.
Let $\tpo(n)$ denote the number of total partial orders over a set with $n$ states.
$\tpo(n) \approx \frac{n!}{2(\ln 2)^{n+1}}$ \cite{BARTHELEMY1980311}, so that we get $\lim_{n \rightarrow \infty} \frac{\sqrt[n]{\tpo(n)}}{n} = \frac{1}{e\ln 2} \approx 0.53$.

\begin{theorem}
If $\mathcal A$ has $n$ states, then the numbers of states of $\U$ and $\B_u$ are $\bigO\big(\tpo(n)\big)$ and $\bigO\big(n\cdot\tpo(n)\big)$, respectively.
\end{theorem}

\begin{proof}
For both automata, the worst case occurs where all states are in the same SCC $C$, say $C = R$.
Starting with $\U$, each macrostate is a tuple $(Q_1, C\setminus Q_1, \preceq, C, D)$.
This means that if we cover the complete statespace, we map each preorder $\preceq'$ to up to $4$ macrostates, namely those with
$C = Q_1 = D$, $C = Q_1 \supsetneq D$, $C \supsetneq Q_1 = D$, and $C \supsetneq Q_1 \supsetneq D$. For $C \supsetneq Q_1$, we remove the minimal elements of $\preceq$ from $Q_1$ and use the set $\preceq = (Q_1 \times Q_1) \cap\preceq'$, while for $C=Q_1$, we use $\preceq = \preceq'$.
For $D \subsetneq Q_1$, $D$ contains all but the maximal (w.r.t.\ $\preceq$) states of $Q_1$.

As this covers all macrostates of $\mathcal U$, $\mathcal U$ has at most $4\cdot\tpo(n)$ macrostates.

For $\B_u$, there are $\bigO(n)$ possible choices for $D$, which leads to $\bigO(n\cdot\tpo(n))$ macrostates.
\end{proof}

\section{Discusion}
\label{sec:conclusion}
We can optimise the construction of $\U$ slightly by moving to transition-based acceptance conditions.
Essentially, where
$(Q_1',Q_2',\preceq',C,\emptyset) \in \trans_u\big((Q_1,Q_2,\preceq,C,D),\sigma\big)$,
$(Q_1',Q_2',\preceq',C,\emptyset)$ would be replaced by $\trans_u\big((Q_1,Q_2,\equiv,C,\emptyset),\sigma\big)$. ($\equiv$ identifies all states it compares; it is the only total preorder acceptable for $D=\emptyset$.)

This is done recursively, until the only macrostates with $D=\emptyset$ left are those with $Q_1 \cap R = \emptyset = Q_2 \cap A$ and (arbitrarily) $C=C_0$.
Note that the initial macrostate has to be changed for this, too.

Removing most macrostates with $D=\emptyset$, this reduces the statespace slightly.
It is also the automaton obtained by de-generalising the standard LTL to transition-based unambiguous generalized \buchi automaton construction.
We can also `re-generalise': every singleton SCC can be removed from the round-robin at the cost of including an individual B\"uchi condition that accepts when the state $s$ is not in $Q_1$ or $Q_2$, respectively, or if $Q_1 \models \trans(s,\sigma)$ or $Q_2 \models \dual{\trans}(s,\sigma)$, respectively, holds. If all components are singleton, we obtain the standard constuction for AVAs / LTL since the preorders of our construction given in Section~\ref{sec:improved-construction} can be omitted.
This way, the $D$ set in a macrostate degenerates to a purely breakpoint construction.
Then, the improved complexity for AVAs matches the current known bounds $n2^n$ for the LTL-to-UBA construction~\cite{DBLP:conf/lics/VardiW86,DBLP:journals/fmsd/JantschMBK21}.



\bibliography{sample}

\begin{thebibliography}{10}

\bibitem{DBLP:books/daglib/0020348}
Christel Baier and Joost{-}Pieter Katoen.
\newblock {\em Principles of model checking}.
\newblock {MIT} Press, 2008.

\bibitem{DBLP:conf/cav/BaierK0K0W16}
Christel Baier, Stefan Kiefer, Joachim Klein, Sascha Kl{\"{u}}ppelholz, David
  M{\"{u}}ller, and James Worrell.
\newblock Markov chains and unambiguous {B{\"{u}}chi} automata.
\newblock In Swarat Chaudhuri and Azadeh Farzan, editors, {\em Computer Aided
  Verification - 28th International Conference, {CAV} 2016, Toronto, ON,
  Canada, July 17-23, 2016, Proceedings, Part {I}}, volume 9779 of {\em Lecture
  Notes in Computer Science}, pages 23--42. Springer, 2016.
\newblock URL: \url{https://doi.org/10.1007/978-3-319-41528-4\_2}, \href
  {http://dx.doi.org/10.1007/978-3-319-41528-4\_2}
  {\path{doi:10.1007/978-3-319-41528-4\_2}}.

\bibitem{BARTHELEMY1980311}
J.P. Barthelemy.
\newblock An asymptotic equivalent for the number of total preorders on a
  finite set.
\newblock {\em Discrete Mathematics}, 29(3):311--313, 1980.
\newblock URL:
  \url{https://www.sciencedirect.com/science/article/pii/0012365X80901594},
  \href {http://dx.doi.org/https://doi.org/10.1016/0012-365X(80)90159-4}
  {\path{doi:https://doi.org/10.1016/0012-365X(80)90159-4}}.

\bibitem{DBLP:conf/tacas/BenediktLW13}
Michael Benedikt, Rastislav Lenhardt, and James Worrell.
\newblock {LTL} model checking of interval markov chains.
\newblock In Nir Piterman and Scott~A. Smolka, editors, {\em Tools and
  Algorithms for the Construction and Analysis of Systems - 19th International
  Conference, {TACAS} 2013, Held as Part of the European Joint Conferences on
  Theory and Practice of Software, {ETAPS} 2013, Rome, Italy, March 16-24,
  2013. Proceedings}, volume 7795 of {\em Lecture Notes in Computer Science},
  pages 32--46. Springer, 2013.
\newblock URL: \url{https://doi.org/10.1007/978-3-642-36742-7\_3}, \href
  {http://dx.doi.org/10.1007/978-3-642-36742-7\_3}
  {\path{doi:10.1007/978-3-642-36742-7\_3}}.

\bibitem{DBLP:conf/ictac/BlahoudekMS19}
Frantisek Blahoudek, Juraj Major, and Jan Strejcek.
\newblock {LTL} to smaller self-loop alternating automata and back.
\newblock In Robert~M. Hierons and Mohamed Mosbah, editors, {\em Theoretical
  Aspects of Computing - {ICTAC} 2019 - 16th International Colloquium,
  Hammamet, Tunisia, October 31 - November 4, 2019, Proceedings}, volume 11884
  of {\em Lecture Notes in Computer Science}, pages 152--171. Springer, 2019.
\newblock URL: \url{https://doi.org/10.1007/978-3-030-32505-3\_10}, \href
  {http://dx.doi.org/10.1007/978-3-030-32505-3\_10}
  {\path{doi:10.1007/978-3-030-32505-3\_10}}.

\bibitem{DBLP:conf/icalp/BokerKR10}
Udi Boker, Orna Kupferman, and Adin Rosenberg.
\newblock Alternation removal in {B{\"{u}}chi} automata.
\newblock In Samson Abramsky, Cyril Gavoille, Claude Kirchner, Friedhelm~Meyer
  auf~der Heide, and Paul~G. Spirakis, editors, {\em Automata, Languages and
  Programming, 37th International Colloquium, {ICALP} 2010, Bordeaux, France,
  July 6-10, 2010, Proceedings, Part {II}}, volume 6199 of {\em Lecture Notes
  in Computer Science}, pages 76--87. Springer, 2010.
\newblock URL: \url{https://doi.org/10.1007/978-3-642-14162-1\_7}, \href
  {http://dx.doi.org/10.1007/978-3-642-14162-1\_7}
  {\path{doi:10.1007/978-3-642-14162-1\_7}}.

\bibitem{DBLP:conf/fossacs/BokerLS22}
Udi Boker, Karoliina Lehtinen, and Salomon Sickert.
\newblock On the translation of automata to linear temporal logic.
\newblock In Patricia Bouyer and Lutz Schr{\"{o}}der, editors, {\em Foundations
  of Software Science and Computation Structures - 25th International
  Conference, {FOSSACS} 2022, Held as Part of the European Joint Conferences on
  Theory and Practice of Software, {ETAPS} 2022, Munich, Germany, April 2-7,
  2022, Proceedings}, volume 13242 of {\em Lecture Notes in Computer Science},
  pages 140--160. Springer, 2022.
\newblock URL: \url{https://doi.org/10.1007/978-3-030-99253-8\_8}, \href
  {http://dx.doi.org/10.1007/978-3-030-99253-8\_8}
  {\path{doi:10.1007/978-3-030-99253-8\_8}}.

\bibitem{Buc62}
J.~Richard B{\"u}chi.
\newblock On a decision method in restricted second order arithmetic.
\newblock In {\em Proc. Int. Congress on Logic, Method, and Philosophy of
  Science. 1960}, pages 1--12. Stanford University Press, 1962.

\bibitem{DBLP:conf/cav/BustanRV04}
Doron Bustan, Sasha Rubin, and Moshe~Y. Vardi.
\newblock Verifying omega-regular properties of {Markov} chains.
\newblock In Rajeev Alur and Doron~A. Peled, editors, {\em Computer Aided
  Verification, 16th International Conference, {CAV} 2004, Boston, MA, USA,
  July 13-17, 2004, Proceedings}, volume 3114 of {\em Lecture Notes in Computer
  Science}, pages 189--201. Springer, 2004.
\newblock URL: \url{https://doi.org/10.1007/978-3-540-27813-9\_15}, \href
  {http://dx.doi.org/10.1007/978-3-540-27813-9\_15}
  {\path{doi:10.1007/978-3-540-27813-9\_15}}.

\bibitem{DBLP:journals/tcs/CartonM03}
Olivier Carton and Max Michel.
\newblock Unambiguous {B{\"{u}}chi} automata.
\newblock {\em Theor. Comput. Sci.}, 297(1-3):37--81, 2003.
\newblock URL: \url{https://doi.org/10.1016/S0304-3975(02)00618-7}, \href
  {http://dx.doi.org/10.1016/S0304-3975(02)00618-7}
  {\path{doi:10.1016/S0304-3975(02)00618-7}}.

\bibitem{DBLP:conf/cav/GastinO01}
Paul Gastin and Denis Oddoux.
\newblock Fast {LTL} to {B{\"{u}}chi} automata translation.
\newblock In G{\'{e}}rard Berry, Hubert Comon, and Alain Finkel, editors, {\em
  Computer Aided Verification, 13th International Conference, {CAV} 2001,
  Paris, France, July 18-22, 2001, Proceedings}, volume 2102 of {\em Lecture
  Notes in Computer Science}, pages 53--65. Springer, 2001.
\newblock URL: \url{https://doi.org/10.1007/3-540-44585-4\_6}, \href
  {http://dx.doi.org/10.1007/3-540-44585-4\_6}
  {\path{doi:10.1007/3-540-44585-4\_6}}.

\bibitem{DBLP:conf/ijcai/GiacomoV13}
Giuseppe~De Giacomo and Moshe~Y. Vardi.
\newblock Linear temporal logic and linear dynamic logic on finite traces.
\newblock In Francesca Rossi, editor, {\em {IJCAI} 2013, Proceedings of the
  23rd International Joint Conference on Artificial Intelligence, Beijing,
  China, August 3-9, 2013}, pages 854--860. {IJCAI/AAAI}, 2013.
\newblock URL:
  \url{http://www.aaai.org/ocs/index.php/IJCAI/IJCAI13/paper/view/6997}.

\bibitem{DBLP:journals/tse/Holzmann97}
Gerard~J. Holzmann.
\newblock The model checker {SPIN}.
\newblock {\em {IEEE} Trans. Software Eng.}, 23(5):279--295, 1997.
\newblock URL: \url{https://doi.org/10.1109/32.588521}, \href
  {http://dx.doi.org/10.1109/32.588521} {\path{doi:10.1109/32.588521}}.

\bibitem{DBLP:journals/fmsd/JantschMBK21}
Simon Jantsch, David M{\"{u}}ller, Christel Baier, and Joachim Klein.
\newblock From {LTL} to unambiguous {B{\"{u}}chi} automata via disambiguation
  of alternating automata.
\newblock {\em Formal Methods Syst. Des.}, 58(1-2):42--82, 2021.
\newblock URL: \url{https://doi.org/10.1007/s10703-021-00379-z}, \href
  {http://dx.doi.org/10.1007/s10703-021-00379-z}
  {\path{doi:10.1007/s10703-021-00379-z}}.

\bibitem{DBLP:conf/icalp/KahlerW08}
Detlef K{\"{a}}hler and Thomas Wilke.
\newblock Complementation, disambiguation, and determinization of {B{\"{u}}chi}
  automata unified.
\newblock In Luca Aceto, Ivan Damg{\aa}rd, Leslie~Ann Goldberg, Magn{\'{u}}s~M.
  Halld{\'{o}}rsson, Anna Ing{\'{o}}lfsd{\'{o}}ttir, and Igor Walukiewicz,
  editors, {\em Automata, Languages and Programming, 35th International
  Colloquium, {ICALP} 2008, Reykjavik, Iceland, July 7-11, 2008, Proceedings,
  Part {I:} Tack {A:} Algorithms, Automata, Complexity, and Games}, volume 5125
  of {\em Lecture Notes in Computer Science}, pages 724--735. Springer, 2008.
\newblock URL: \url{https://doi.org/10.1007/978-3-540-70575-8\_59}, \href
  {http://dx.doi.org/10.1007/978-3-540-70575-8\_59}
  {\path{doi:10.1007/978-3-540-70575-8\_59}}.

\bibitem{DBLP:conf/atva/KarmarkarJC13}
Hrishikesh Karmarkar, Manas Joglekar, and Supratik Chakraborty.
\newblock Improved upper and lower bounds for {B{\"{u}}chi} disambiguation.
\newblock In Dang~Van Hung and Mizuhito Ogawa, editors, {\em Automated
  Technology for Verification and Analysis - 11th International Symposium,
  {ATVA} 2013, Hanoi, Vietnam, October 15-18, 2013. Proceedings}, volume 8172
  of {\em Lecture Notes in Computer Science}, pages 40--54. Springer, 2013.
\newblock URL: \url{https://doi.org/10.1007/978-3-319-02444-8\_5}, \href
  {http://dx.doi.org/10.1007/978-3-319-02444-8\_5}
  {\path{doi:10.1007/978-3-319-02444-8\_5}}.

\bibitem{DBLP:journals/tocl/KupfermanV01}
Orna Kupferman and Moshe~Y. Vardi.
\newblock Weak alternating automata are not that weak.
\newblock {\em {ACM} Trans. Comput. Log.}, 2(3):408--429, 2001.
\newblock URL: \url{https://doi.org/10.1145/377978.377993}, \href
  {http://dx.doi.org/10.1145/377978.377993} {\path{doi:10.1145/377978.377993}}.

\bibitem{DBLP:conf/cav/KwiatkowskaNP11}
Marta~Z. Kwiatkowska, Gethin Norman, and David Parker.
\newblock {PRISM} 4.0: Verification of probabilistic real-time systems.
\newblock In Ganesh Gopalakrishnan and Shaz Qadeer, editors, {\em Computer
  Aided Verification - 23rd International Conference, {CAV} 2011, Snowbird, UT,
  USA, July 14-20, 2011. Proceedings}, volume 6806 of {\em Lecture Notes in
  Computer Science}, pages 585--591. Springer, 2011.
\newblock URL: \url{https://doi.org/10.1007/978-3-642-22110-1\_47}, \href
  {http://dx.doi.org/10.1007/978-3-642-22110-1\_47}
  {\path{doi:10.1007/978-3-642-22110-1\_47}}.

\bibitem{DBLP:journals/tcs/MiyanoH84}
Satoru Miyano and Takeshi Hayashi.
\newblock Alternating finite automata on omega-words.
\newblock {\em Theor. Comput. Sci.}, 32:321--330, 1984.
\newblock URL: \url{https://doi.org/10.1016/0304-3975(84)90049-5}, \href
  {http://dx.doi.org/10.1016/0304-3975(84)90049-5}
  {\path{doi:10.1016/0304-3975(84)90049-5}}.

\bibitem{DBLP:journals/tcs/MullerSS92}
David~E. Muller, Ahmed Saoudi, and Paul~E. Schupp.
\newblock Alternating automata, the weak monadic theory of trees and its
  complexity.
\newblock {\em Theor. Comput. Sci.}, 97(2):233--244, 1992.
\newblock URL: \url{https://doi.org/10.1016/0304-3975(92)90076-R}, \href
  {http://dx.doi.org/10.1016/0304-3975(92)90076-R}
  {\path{doi:10.1016/0304-3975(92)90076-R}}.

\bibitem{DBLP:conf/litp/MullerS84}
David~E. Muller and Paul~E. Schupp.
\newblock Alternating automata on infinite objects, determinacy and rabin's
  theorem.
\newblock In Maurice Nivat and Dominique Perrin, editors, {\em Automata on
  Infinite Words, Ecole de Printemps d'Informatique Th{\'{e}}orique, Le Mont
  Dore, France, May 14-18, 1984}, volume 192 of {\em Lecture Notes in Computer
  Science}, pages 100--107. Springer, 1984.
\newblock URL: \url{https://doi.org/10.1007/3-540-15641-0\_27}, \href
  {http://dx.doi.org/10.1007/3-540-15641-0\_27}
  {\path{doi:10.1007/3-540-15641-0\_27}}.

\bibitem{RohdePhD97}
Gareth~Scott Rohde.
\newblock {\em Alternating automata and the temporal logic of ordinals}.
\newblock PhD thesis, University of Illinois at Urbana-Champaign, 1997.

\bibitem{DBLP:conf/stacs/Schewe09}
Sven Schewe.
\newblock B{\"{u}}chi complementation made tight.
\newblock In Susanne Albers and Jean{-}Yves Marion, editors, {\em 26th
  International Symposium on Theoretical Aspects of Computer Science, {STACS}
  2009, February 26-28, 2009, Freiburg, Germany, Proceedings}, volume~3 of {\em
  LIPIcs}, pages 661--672. Schloss Dagstuhl - Leibniz-Zentrum f{\"{u}}r
  Informatik, Germany, 2009.
\newblock URL: \url{https://doi.org/10.4230/LIPIcs.STACS.2009.1854}, \href
  {http://dx.doi.org/10.4230/LIPIcs.STACS.2009.1854}
  {\path{doi:10.4230/LIPIcs.STACS.2009.1854}}.

\bibitem{DBLP:journals/corr/Vardi11}
Moshe~Y. Vardi.
\newblock The rise and fall of {LTL}.
\newblock In Giovanna D'Agostino and Salvatore~La Torre, editors, {\em
  Proceedings of Second International Symposium on Games, Automata, Logics and
  Formal Verification, GandALF 2011, Minori, Italy, 15-17th June 2011}, 2011.
\newblock invited talk.
\newblock URL: \url{https://www.cs.rice.edu/~vardi/papers/gandalf11-myv.pdf}.

\bibitem{DBLP:conf/lics/VardiW86}
Moshe~Y. Vardi and Pierre Wolper.
\newblock An automata-theoretic approach to automatic program verification
  (preliminary report).
\newblock In {\em Proceedings of the Symposium on Logic in Computer Science
  {(LICS} '86), Cambridge, Massachusetts, USA, June 16-18, 1986}, pages
  332--344. {IEEE} Computer Society, 1986.

\end{thebibliography}

\newpage
\appendix

\section{Proof of Lemma~\ref{lem:unique-distance}}
\label{app:unique-distance}
\uniqueDistance*

\begin{proof}
Intuitively, the distance function is to define a minimal number of steps to escape from rejecting SCCs over different accepting run DAGs and maximal over different branches of one such run DAG.
We first show that such a sequence of distance function exists and then prove that it is unique.

Let $C$ be a rejecting SCC of $\A$; the proof for the case for a rejecting SCC of $\dual{\A}$ is similar.
Below, we describe how to obtain a sequence of distance values for each state $q \in C \cap Q^i_1$ with $i \geq 0$ in order to form a consistent sequence $\Phi_w$.
For $q \in C \cap Q^i_1$ at the level $i$, we first obtain an accepting run DAG $\G_{\wordletter{w}{i\cdots}}$ over $\wordletter{w}{i\cdots}$ starting from $\vertex{q}{0}$.
One can define the maximal distance, say $K$, over all branches from $\vertex{q}{0}$ to escape the rejecting SCC $C$.
Such a maximal distance value must exist and be a finite value, since all branches will eventually get trapped in accepting SCCs.
For all accepting run DAGs $\G'_{\wordletter{w}{i\cdots}}$ over $\wordletter{w}{i\cdots}$ starting from the vertex $\vertex{q}{0}$, there are only finitely many run DAGs of depth $K$ from the vertex $\vertex{q}{0}$;
we denote the finite set of such run DAGs of depth $K$ by $P_{q,i}$.
We then denote the maximal distance over one finite run DAG $G_{q,i, K} \in P_{q,i}$ by $K_{G_{q,i, K}}$.
(Note that we set the distance to $\infty$ for a finite branch in $G_{q,i, K}$ if it does not visit a state outside $C$.)
We then set $d_i(q) = \min \setnocond{ K_{G_{q,i, K}}: G_{q,i, K} \in P_{q,i}} \leq K$.
One of $G_{q,i, K}$ must provide the \emph{minimal} value, so that $d_i(q)$ is well defined.
This way, we can define the sequence of distance functions $\textbf{d} = d_0 d_1 \cdots$ for the sequence $\R_w$.

We can show that the sequence $\R_w \times \textbf{d}$ is consistent by induction on all the distance value $k > 0$.
We only prove the case for a state $q \in R\cap Q^i_1$ that belongs to a rejecting SCC in $\A$.
The proof for a state $q \in A \setminus Q^i_1$ is similar.
We first prove the \textbf{induction basis} when $\mathbf{k = 1}$.

Let $q \in C \cap Q^i_1$ be a state with $d_i(q) = 1$.
By definition, we know that $K \geq 1$.
Moreover, there must be a run DAG $G_{q, i, K}$ of depth $K$ as part of an accepting run DAG $\G_{\wordletter{w}{i\cdots}}$ in which the level $1$ only contains the states in $S \subseteq Q \setminus C$ such that $S \models \trans(q, \wordletter{w}{i})$.
Since $G_{q, i, K}$ is part of an accepting run DAG over $\wordletter{w}{i\cdots}$, we also have that $\wordletter{w}{i+1\cdots} \in \lang{\A^s}$ for all $s \in S$.
Hence, $S \subseteq Q^{i+1}_1$, and further $S \subseteq Q^{i+1}_1 \setminus C$.
It immediately follows that $Q^{i+1}_1\setminus C \models \trans(q, \wordletter{w}{i})$, in compliance with the rules R1a and R1b.

Now we prove the \textbf{induction step} ($\mathbf{k\mapsto k+1}$).
Assume that $d_i(q) = k + 1$ and for all distance values $k' \leq k$, the distance function $\textbf{d} = d_0 d_1 \cdots$ is consistent.
Again, we should be able to find the run DAG $G_{q, i, k+1}$ of depth $k+1$ of $\G_{\wordletter{w}{i\cdots}}$ in which $S$ is the set of states in level $1$.
Obviously, $S \models \trans(q, \wordletter{w}{i})$.
Similarly, $S \subseteq Q^{i+1}_1$ holds.
For all states $p \in S \cap C$, we have that $d_{i+1}(p) \leq k = d_i(q) - 1$ (as witnessed by the run DAG $G_{q, i+1, k}$ over $\wordletter{w}{i\cdots}$ obtained from $G_{q, i, k+1}$ by removing level $0$).
Thus, we have $S\cap C \subseteq \setcond{p \in C \cap Q^{i+1}_1}{d_{i+1}(p) \leq d_i(q) - 1}$.
Together with the fact that $S\setminus C \subseteq Q^{i+1}_1 \setminus C$, we have that $ Q^{i+1}_1\setminus C \cup \setcond{p \in C \cap Q^{i+1}_1}{d_{i+1}(p) \leq d_i(q) - 1} \models \trans(q, \wordletter{w}{i})$, in compliance with R1a.

We can prove R1b easily by contraposition.
Assume that $Q^{i+1}_1\setminus C \cup \setcond{p \in C\cap Q^{i+1}_1}{d_{i+1}(p) \leq d_{i}(q) - 2} \models \trans(q, \wordletter{w}{i})$.
Then, there exists a run DAG $G'_{q,i,K}$ in which the level $1$ contains all the states in $Q^{i+1}\setminus C \cup \setcond{p \in C\cap Q^{i+1}_1}{d_{i+1}(p) \leq d_{i}(q) - 2} $.
Since $\textbf{d}$ is consistent when the distance value is not greater than $k$, so $K'_{G_{q, i, K}} \leq k$ by induction hypothesis.
Thus, by definition, we should have $d_i(q) = k$, leading to contradiction.

Therefore, $\R_w \times \textbf{d}$ is a consistent sequence.

Now we prove that the distance function $\textbf{d}$ is unique to $\R_w$.
We observe that a consistent sequence $ \textbf{c} = c_0 c_i \ldots$ will provide an accepting run DAG for all tails $\wordletter{w}{i\cdots}$ with $i \in \naturals$: by always choosing the satisfying sets from R1a and R2a, respectively, from a state $q$ in the domain of $c_i$, we will leave its SCC $C$ from level $i$ of the run DAG in $c_i(q)$ steps, so that no run can get stuck in a rejecting SCC.

This also provides $c_i(q)\geq d_i(q)$ for all $i\in \naturals$ and all $q$ in their domain, by definition of $\textbf{d}$.

We now show by induction that, for all $k>0$ and all $i \in \naturals$, the pre-image of $c_i$ and $d_i$ for $k$ coincide.
\medskip

The \textbf{induction basis} is the case of $k=1$, and thus $d_i(q)=1$.
For this to happen, it requires that $C$ can be left immediately, which would then allow for using rule R1b or R2b, as the left set of the union alone suffices for satisfaction.
$c_i(q)=1$ is therefore the only possible assignment (and in compliance with rules R1a and R1b).
\medskip

The \textbf{induction step} is from $k$ to $k+1$.

Let $d_i(q)=k+1$. We have already shown $c_i(q) \geq k+1$.

By definition, there is an accepting run DAG from $q$ at level $i$, such that $C$ is left in $k+1$ steps.
We fix such a run DAG. We observe that for every successor $s$ in level $i+1$ we have that it is either outside of $C$, or $d_{i+1}(s)\leq k$.
Using the induction hypothesis, the latter entails $c_{i+1}(s) \leq k$.
Therefore, rule R1a or R2a applies.
We now assume for contradiction that the respective rule R1b or R2b does not apply.
But then we can satisfy $\trans(p,\wordletter{w}{i})$ or $\dual{\trans}(p,\wordletter{w}{i})$, respectively, by only those states not in $C$ or with $d_{i+1}(s)<k$, which would entail $d_i(q)\leq k$ (by making such a choice and inserting the witnessing run graphs in level $i+1$.
This closes the contradiction, and provides $c_i(q)=k+1$.

This completes the induction and provides the Lemma.
\end{proof}

\section{Proof of Theorem~\ref{thm:correctness-new-construction}}
\label{app:new-construction-correctness}
\newConstruction*

\begin{proof}
We first observe that, for an accepting macrorun
$\rho = (Q_1^0,Q_2^0,\{\preceq_C^0\}_{C \in \mathcal S},S^0,D^0)\linebreak
(Q_1^1,Q_2^1,\{\preceq_C^1\}_{C \in \mathcal S},S^1,D^1)
(Q_1^2,Q_2^2,\{\preceq_C^2\}_{C \in \mathcal S},S^2,D^2)\cdots$ on a word $w$ we have that

\begin{enumerate}
    \item $q \in Q_1^i$ implies $w[i\cdots] \in \lang{\A^q}$ and
    \item $q \in Q_2^i$ implies $w[i\cdots] \in \lang{\widehat{\A}^q}$.
\end{enumerate}

To show (1), we observe that we can produce a run DAG $\G_{q,i}$ of $\A^q$ on $w[i\cdots]$ by selecting, for all $j \geq i$ and all $q' \in Q_1^j \cap C$ for some rejecting SCC $C \subseteq R$ only successors for the minimal (w.r.t.\ $\preceq_C^{j+1}$) downward closed set $D' \subseteq C \cap Q_1^{j+1}$ such that $D' \cup (Q_1^{j+1}\setminus C) \models \trans(q',w[j])$.

We now show that all the branches in the run DAG cannot get stuck in $C$.
As the macrorun $\rho$ is accepting, there must be a next time $k > j$ where either $Q^k_1 \cap C =\emptyset$ (which trivially means that the branches do not get stuck in $C$) or $D^k = Q^k_1 \cap C$---either happens at the latest after $|\mathcal S|$ accepting macrostates have been visited.
Recall that in the construction, when $D^{k-1} = \emptyset$ (and thus a visit to accepting macrostate) and $C = S^{k} = \nxt(S^{k-1})$, we have that $D^k = Q^k_1 \cap C$ as $C \subseteq R$ is a non-accepting SCC.
For the latter case, it is to show by induction that all branches in the run DAG originating from $q$ are henceforth either not in $C$, or in $D$, so that $C$ is left at the latest when the $|\mathcal S|+1^{st}$ accepting macrostate is visited.
The reason why the branches are in $D$ is that according to the construction, we only leave the smallest downward closed set of successors for $D$ in $D'$.
Since $\rho$ visits empty $D$-sets for infinitely many times, the run DAG must not be stuck in $C$ for all $C \in \mathcal{S}$.
Therefore, the run DAG $\G_{q,i}$ is accepting.
It follows that $\wordletter{w}{i\cdots} \in \lang{\A^q}$.

The proof for (2) is similar.

Using this, we first obtain $\lang{\B_u} \subseteq \lang{\A}$ (as an accepting macrorun must satisfy $\iota \in Q_1^0$).

Second, it implies that $\R_w = Q_1^0Q_1^1Q_1^2\cdots$ holds for all accepting macroruns.

With Corollary \ref{cor:unique-preorder} and the observation that the update of the last two components of a macrostate ($S$ and $D$) are deterministic, this entails unambiguity (there is at most one accepting macrorun).
Note that any wrong guesses for preorders will violate the local consistency rules and those macroruns will therefore discontinue the moment violations are found.

Finally, if $w\in \lang{\A}$, then we have for the unique sequence $\R_w = Q_1^0Q_1^1Q_1^2\cdots$ that $\iota \in Q_1^0$, and we can use Lemma \ref{lem:unique-distance} to construct the corresponding unique distance functions $d_0d_1\ldots$.

Now we show how to construct an accepting macrorun $\rho = (Q_1^0,Q_2^0,\{\preceq_C^0\}_{C \in \mathcal S},S^0,D^0)\linebreak
(Q_1^1,Q_2^1,\{\preceq_C^1\}_{C \in \mathcal S},S^1,D^1)
(Q_1^2,Q_2^2,\{\preceq_C^2\}_{C \in \mathcal S},S^2,D^2)\cdots$ of $\B_u$ over $w$ where $Q^i_2 = Q\setminus Q^i_1$ and $Q^0_1 Q^1_1 Q^2_1\cdots$ is of course the unique sequence $\R_w$;
We will set the preorders $\setnocond{\preceq^i_C}_{C \in \mathcal{S}}$ as defined in Lemma~\ref{lem:exist-preorders}, based on the distance functions $d_0 d_1 \cdots$.
The updates of $S^i$ and $D^i$ are then deterministic with respect to $Q^i_1$ and $\setnocond{\preceq^i_C}_{C \in \mathcal{S}}$.
Apparently, the preorders meet all the local consistency constraints, according to Lemma~\ref{lem:exist-preorders}.
So, the macrorun $\rho$ is of infinite length.

By Lemma~\ref{lem:one-one-mapping} and Corollary~\ref{cor:unique-preorder}, we also know that it is the unique preorder sequence for $\R_w$, which gives $q \preceq^i_C q' \Longleftrightarrow d_i(q) \leq d_i(q')$ and $q \prec^i_C q' \Longleftrightarrow d_i(q) < d_i(q')$ for all $i \in \naturals$.
With this, by Definition~\ref{def:consistent-distances} and Definition~\ref{def:consistent-preorders}, it is now easy to show with an inductive argument similar to the one in Lemma \ref{lem:one-one-mapping} that, if $D^i \neq \emptyset$, $\sup\{d_i(q) \mid q \in D^i\} = \sup\{d_{i+1}(q) \mid q \in D^{i+1}\} +1$ (choosing $\sup \emptyset = 0$).
Since all the distance values of the states in $D^i \neq \emptyset$ are finite and the maximal value in $D^i$ is decreasing, the value will eventually become $0$.
In other words, for every $i > 0$ with $D^i \neq \emptyset$, there will be some $j > i$ such that $D^{j} = \emptyset$.
Thus, the macrorun $\rho$ must be accepting.

We therefore also have $\lang{\A}\subseteq \lang{\B_u}$.
\end{proof}

\section{Proof of Theorem~\ref{thm:correctness-improved-construction}}
\label{app:improved-construction}
\begin{proof}

We first observe that, for an accepting macrorun
$\rho = (Q_1^0,Q_2^0,\preceq_{C_0}^0,C^0,D^0)\linebreak
(Q_1^1,Q_2^1,\preceq_{C_1}^1,C^1,D^1)
(Q_1^2,Q_2^2,\preceq_{C_2}^2,C^2,D^2)\cdots$ on a word $w$ we have that

\begin{enumerate}
    \item $q \in Q_1^i$ implies $w[i\cdots] \in \A^q$ and
    \item $q \in Q_2^i$ implies $w[i\cdots] \in \widehat{\A}^q$,
\end{enumerate}
with exactly the same proof as in Theorem \ref{thm:correctness-new-construction}.

Using this, we first similarly obtain $\lang{\U} \subseteq \lang{\A}$ (as an accepting macrorun must satisfy $\iota \in Q_1^0$) and that $\R_w = Q_1^0Q_1^1Q_1^2\cdots$ holds for all accepting macroruns.

Next we show that $\U$ can simulate $\B_u$:
if
$\rho = (Q_1^0,Q_2^0,\{\preceq_C^0\}_{C \in \mathcal S},S^0,D^0)
(Q_1^1,Q_2^1,\{\preceq_C^1\}_{C \in \mathcal S},S^1,$ $D^1)
(Q_1^2,Q_2^2,\{\preceq_C^2\}_{C \in \mathcal S},S^2,D^2)\cdots$ is an accepting macrorun of $\B_u$ on a word $w$, then $\U$ has an accepting macrorun
$\dual{\rho} = (Q_1^0,Q_2^0,\preceq_0,S^0,D^0)
(Q_1^1,Q_2^1,\preceq_1,S^1,D^1)
(Q_1^2,Q_2^2,\preceq_2,S^2,D^2)\cdots$, where
\begin{itemize}
    \item if $S^i \subseteq R$, then $\preceq_i$ is a total preorder on $S^i \cap Q_1^i$ such that
    \begin{itemize}
        \item $\preceq_i = \preceq_{S^i}^i$ if $D^i=S^i \cap Q_1^i$ and
        \item otherwise, the maximal elements of $\preceq_i$ are set to $(S^i \cap Q_1^i) \setminus D^i$, and the restriction of $\preceq_i$ to $D^i \times D^i$ agrees with the restriction of $\preceq_{S^i}^i$ to $D^i \times D^i$, and
    \end{itemize}
    \item similarly, if $S^i \subseteq A$, then $\preceq_i$ is a total preorder on $S^i \cap Q_2^i$ such that
    \begin{itemize}
        \item $\preceq_i = \preceq_{S^i}^i$ if $D^i=S^i \cap Q_2^i$ and
        \item otherwise, the maximal elements of $\preceq_i$ are set to $(S^i \cap Q_2^i) \setminus D^i$, and the restriction of $\preceq_i$ to $D^i \times D^i$ agrees with the restriction of $\preceq_{S^i}^i$ to $D^i \times D^i$.
    \end{itemize}
\end{itemize}

Let $m_i = \sup\{d_i(q) \mid q \in D^i\}$.
Intuitively, the total preorder $\preceq_i$ simply orders those states in $s \in S^i \cap Q^i_1$ resp.\ $s \in S^i\cap Q^i_2$ with $d_i(s)\leq m_i$ correctly, while aggregating all such states $s$ with $d_i(s) > m_i$
as maximal elements.
It is easy to extend the proof of Theorem \ref{theo:Bu} to show that this satisfies all local constraints for Rule R1' resp.\ R2'.
Note that our preorders $\preceq_i$ are no longer defined over all SCCs, so Lemma~\ref{lem:one-one-mapping} may not entirely hold here.
Now we show that $(Q^{i+1}_1, Q^{i+1}_2, \preceq_{i+1}, S^{i+1}, D^{i+1})$ is a valid $\wordletter{w}{i}$-successor of $(Q^{i}_1, Q^{i}_2, \preceq_{i}, S^{i}, D^{i})$.
First, the local consistency for the reachable states $Q^{i+1}_1$ and $Q^{i+1}_2$ clearly holds since $\rho$ also visits the same set of reachable states.
If $D^i = \emptyset$, two constructions behave the same.
So we only need to show it is valid when $D^i \neq \emptyset$.
We next show that the requirements of Rule R1' are met;
the proof for Rule R2' is similar.
If $D^i = S^i \cap Q^i_1$, then $D^{i+1} $ is the smallest downward closed set w.r.t. $\preceq_{i+1}$ such that $D^{i+1} \cup (Q^{i+1}_1 \setminus S^{i + 1}) \models \land_{s \in D^i}\trans(s, \sigma)$.
If $D^{i+1} = S^{i+1} \cap Q^{i+1}_1$, then $\preceq_{i+1} = \preceq^{i+1}_{S^i}$, the consistency clearly holds.
If $D^{i+1} \subset S^{i+1} \cap Q^{i+1}_1$, then, for every pair of states $q, q'$ with $q \preceq^i_{S^i} q'$, there must be a state $r \in D^{i+1}$ satisfying Definition~\ref{def:consistent-preorders}, since $D^{i+1} \cup (Q^{i+1}_1 \setminus S^{i + 1}) \models \land_{s \in D^i}\trans(s, \sigma)$ and $\preceq_{i+1} = \preceq^{i+1}_{S^{i}}$ over $D^{i+1} \times D^{i+1}$, where $S^{i+1} = S^i$.
If $D^i \neq S^i \cap Q^i_1$, then we have $D^i \subset S^i \cap Q^i$.
For states $q, q' \in D^i$ with $q \preceq_i q'$, the proof is similar.
Consider $q \in D^i, q' \in (S^i \cap Q^i_1)\setminus D^i$ with $q \prec_i q'$:
it is impossible that $D^{i+1} = S^{i+1} \cap Q^{i+1}_1$.
This is because that since $\rho[i]$ and $\rho[i+1]$ are consistent sequence, they $D^i$ will include all states from $S^i \cap Q^i_1$, violating the assumption and Definition~\ref{def:consistent-preorders}.
So, it must be the case that $D^{i+1} \subset S^{i+1} \cap Q^{i+1}_1$.
Then we can just select the $r$-state of Definition~\ref{def:consistent-preorders} as a minimal element $(S^{i+1} \cap Q^{i+1}_1)\setminus D^{i+1}$, satisfying R1' in Definition~\ref{def:consistent-preorders}.
Hence, the macrorun $\dual{\rho}$ is infinite and visits infinitely many empty $D$-sets.

This provides $\lang{\U}\supseteq\lang{\B_u}$.
With $\lang{\B_u}=\lang{\A}$ (Theorem \ref{theo:Bu}), we now have $\lang{\U}=\lang{\A}$.

To show that there is only one accepting macrorun, we turn the argument of assigning values around.
Our proof idea is to establish some properties of every accepting macrorun in $\U$ and prove that there is only one macrorun satisfying such properties.

We have already established that  $\R_w = Q_1^0Q_1^1Q_1^2\cdots$ holds, and will use the unique extension  $\Phi_w = (Q_1^0, d_0) (Q_1^1, d_1)\cdots$ to distance functions (Lemma \ref{lem:unique-distance}).

Let $\dual{\rho} = (Q_1^0,Q_2^0,\preceq_0,C^0,D^0)
(Q_1^1,Q_2^1,\preceq_1,C^1,D^1)
(Q_1^2,Q_2^2,\preceq_2,C^2,D^2)\cdots$ be an accepting macrorun of $\mathcal U$ on a word $w$.
Let $i>0$ be an accepting macrostate position in $\dual{\rho}$, and let $i'<i$ be the last accepting macrostate position that occurred before $i$.

We assume $C^i \subseteq R$, the case $C^i \subseteq A$ is entirely similar.
We note that $C^j = C^i$ for all $i' < j \leq i$.
Hence, in the following, we actually work on the SCC $C^i$.

We now show by induction over $j$ that, for all $i'<j\leq i$
we have that, for $m_j = \sup\{d_j(q) \mid q \in D^j\}$, the total preorder $\preceq_j$ simply orders those states in $s \in C^i \cap Q_1^j$ correctly, while the remaining states are maximal elements of $\preceq_j$:
\begin{enumerate}
    \item for all $q \in D^j$ and all $ q' \in C^i \cap Q_1^j.\ d_j(q)\leq d_j(q') \Leftrightarrow q \preceq_j q'$, and $d_j(q)\leq i-j$ hold,
    \item for all $q \in C^i \cap Q_1^j$ and all $ q' \in (C^i \cap Q_1^j) \setminus D^j.\  q \preceq_j q'$ and $d_j(q') > i-j$ hold, and
    \item $m_j = \sup\{d_j(q) \mid q \in D^j\} = i-j$, using $\sup\emptyset = 0$.
\end{enumerate}

For the \textbf{induction basis}, this is true by definition for $j=i$.
By assumption, $D^j = \emptyset$ and by Definition~\ref{def:separated-ba}, $\preceq_j$ identifies only one equivalence class---the maximal equivalence class, since $D^j$ is the smallest downward closed set w.r.t. $\preceq_j$ such that $D^j \cup (Q^{j} \setminus C^i) \models \land_{s\in D^{j-1}} \trans(s, \wordletter{w}{i})$.
By definition, either $D^j = Q^j \cap C^i$ or $D^j = (Q^j \cap C^i)\setminus M^j$ holds where $M^j$ is the maximal elements of $\preceq_j$.
Hence, for all $q \in C^i\cap Q^j_1, q' \in (C^i\cap Q^j_1)\setminus D^j = C^i\cap Q^j_1$, $q \simeq_j q'$ and thus also $q \preceq_j q'$;
By definition, $d_j(q') > 0$ holds always.
Therefore, Item (2) holds.
Moreover, Item (3) clearly holds.
For Item (1), since $q$ does not exist, we simply say Item (1) is true for technical reason.

For the \textbf{induction step} $j \mapsto j-1$ (assuming $j>i'+1$),
Rules R1 and R1' imply with (1) and (2) from the induction hypothesis that (1) and (2) also hold for $j-1$.

For $j=i $, and thus $D^{j} = \emptyset$, by Rule R1 it is exactly those states with $d_{j-1}(q) = 1$ that are in $D^{j-1}$.
Assume that $D^{j-1} \neq \emptyset$.
Clearly, for all $q \in D^{j-1}$, $d_{j-1}(q) = 1 \leq i -(j-1) = 1$.
Hence, Item (3) holds.
For $q' \in D^{j-1} \subseteq C^i \cap Q^{j-1}_1$, we have $d_{j-1}(q) = d_{j-1}(q')$ hold.
Further, $q \simeq_{j-1} q'$ holds since if there is $q'$ such that $q \prec_j q'$, then the local consistency of Definition~\ref{def:consistent-preorders} will be violated because $\emptyset \cup (Q^j_1 \setminus C^i ) \models \trans(q', \wordletter{w}{j-1})$.
If $q' \in (C^i \cap Q^{j-1}_1) \setminus D^{j-1}$, by definition, $q \prec_{j-1} q'$ and thus also $q \preceq_{j-1} q'$ since $q'$ is a maximal element w.r.t.\ $\preceq_{j-1}$.
Moreover, by R1', there exists an $r$-state for $q$ and $q'$ satisfying Definition~\ref{def:consistent-preorders}.
If $d_{j-1}(q') \leq i - (j - 1) = 1$, i.e., $d_{j-1}(q') = 1 = d_{j-1}(q)$, this violates the existence of $r$-state in Definition~\ref{def:consistent-preorders}, according to Definition~\ref{def:consistent-distances}.
Thus, $d_{j-1}(q') > i - (j-1) = 1$.
It follows that Item (2) holds.
Now we only need to prove that $d_{j-1}(q) \leq d_{j-1}(q') \Leftrightarrow q \preceq_{j-1} q'$ and $d_{j-1}(q) \leq i - (j-1)$ hold when $q' \in (C^i \cap Q^{j-1}_1) \setminus D^{j-1}$.
The case when $q' \in  D^{j-1}$ has already been proved above.
By Item (2), we already have $d_{j-1}(q) < d_{j-1}(q), q \prec_{j-1} q'$ and clearly, $d_{j-1}(q) \leq 1$ hold.
Hence, Item (1) holds as well.
It follows that when $j = i$, the three items also hold when $j \mapsto j -1$.
If there are no such states, i.e. $D^{j-1} = \emptyset$, then the backwards deterministic definition of $(Q_1^{j-1},\preceq_{j-1})$ from $(Q_1^{j},\preceq_{j})$ according to Rule R1'
implies\footnote{note that $D^j = \emptyset$, so that all states in $C^i \cap Q^j$ are maximal elements of, and therefore identified by $\preceq_j$}
(with the absence states $d_{j-1}(q) = 1$ and Rule R1) that all states in $C^i \cap Q_1^{j-1}$ are identified by $\preceq_{j-1}$ and $D^{j-1}=\emptyset$.
Such a macrostate is accepting, which contradicts $j>i'+1$.

For $j<i$, we first observe that, for all states $q \in C^i \cap Q^{j-1}_1$, $d_{j-1}(q)\leq i-j+1$ holds iff $q \in D^{j-1}$ with the same backwards deterministic argument as above (using Rules R1 and R1') from the induction hypothesis.
But we also have to establish that there is a state $q \in D^{j-1} \subseteq C^i \cap Q^{j-1}_1$ with $d_{j-1}(q) = i-j+1$.

We assume for contradiction that this is not the case. Then $\sup\{d_{j-1}(q) \mid q \in D^{j-1}\} \leq i-j$, which implies that the set $\dual{D}^j = \{q \in C^i \cap Q^j_1 \mid d_j(q) < i-j\}$ is a downwards closed (w.r.t.\ $\preceq_j$) set, which is strictly smaller than $D^j$ and satisfies the other transition requirements.
Therefore $D^j$ does not satisfy the minimality requirement (contradiction).

The proof for the three items are then easy.
First, Item (3) has been proved above.
By induction hypothesis, we have that the three items hold on position $j$.
We now prove Item (2).
For all $q \in C^i \cap Q^{j-1}_1$ and $q' \in M^{j-1} = (C^i \cap Q^{j-1}_1)\setminus D^{j-1}$ (if it exists), by definition $q'$ is a maximal element w.r.t.\ $\preceq_{j-1}$.
Clearly, $q \prec_{j-1} q'$ and thus $q \preceq_{j-1} q'$.
Suppose $d_{j-1}(q') \leq i- j + 1$.
But then we have a state $q \in D^{j-1}$ such that $d_{j-1}(q) = i - j + 1$.
Since $q \prec_{j-1} q'$, there must exist an $r$-state in $C^i \cap Q^j_1$ satisfying R1' of Definition~\ref{def:consistent-preorders}.
By Definition~\ref{def:consistent-distances}, $d_j(r) \leq i-j$.
We then have that $r \in D^j$ because there is a state $r' \in D^j$ such that $d_j(r') = i-j$ and then we have $r \preceq_{j} r'$ by Item (1).
(If $D^j = \emptyset$, it immediately leads to contradiction.)
But then, $\setcond{p \in C^i \cap Q^j}{ p \prec_j r} \cup (Q^j_1\setminus C^i) \not\models \trans(q, \wordletter{w}{j-1})$ since $\setcond{p \in C^i \cap Q^j}{ p \prec_j r} = \setcond{p \in C^i \cap Q^j}{ d_j(p) \leq d_j( r) - 1 < i - j = d_{j-1}(q) - 1}$ (by induction hypothesis), violating Definition~\ref{def:consistent-distances}.
It follows that $d_{j-1}(q') > i - j + 1$.
Therefore, Item (2) holds.
Item (1) can be proven similarly.
One can also prove similarly the following when $C^i \subseteq A$:
\begin{enumerate}
    \item for all $q \in D^j$ and all $ q' \in C^i \cap Q_2^j.\ d_j(q)\leq d_j(q') \Leftrightarrow q \preceq_j q'$, and $d_j(q)\leq i-j$ hold,
    \item for all $q \in C^i \cap Q_2^j$ and all $ q' \in (C^i \cap Q_2^j) \setminus D^j.\  q \preceq_j q'$ and $d_j(q') > i-j$ hold, and
    \item $m_j = \sup\{d_j(q) \mid q \in D^j\} = i-j$, using $\sup\emptyset = 0$.
\end{enumerate}

This closes the inductive argument.
\medskip

Finally, we observe that the simulation macrorun $\dual{\rho}$ for the sole accepting macrorun of $\B_u$ is the only macrorun that satisfies the three item requirements, based on following facts.
(i) There is only a single initial macrostate $(Q^0_1, Q^2_2, \preceq_0, C^0, D^0)$ that fits $\R_w$ (with all states in $C_0 \cap Q^0_1$ or ($C_0 \cap Q^0_2$) being maximal w.r.t. $\preceq^0_C$ since $D^0 = \emptyset$), and when we take a transition from an accepting macrostate $(Q^j_1, Q^j_2, \preceq_j, C^j, D^j = \emptyset)$ (including the first), the next SCC $C^{j+1} = \nxt(C^j)$ is deterministically selected.
(ii) Moreover, all relevant states from the SCC $C^{j+1} \cap Q^{j+1}_1$ (resp. $C^{j+1}\cap Q^{j+1}_2$) are in the $D^{j+1}$ component, since $D^{j+1} = C^{j+1}\cap Q^{j+1}_1$ (resp. $D^{j+1} = C^{j+1}\cap Q^{j+1}_2$) by construction.
We have seen that $\sup \setcond{d_j(q)}{q \in C^j\cap Q^j_1}$ (resp. $\sup \setcond{d_j(q)}{q \in C^j\cap Q^j_2}$) determines the distance to the next breakpoint, by Item (3) when $D^{j-1} = \emptyset$ for all $j > 0$, and thus the $\preceq_j$ and $D^j$ up to the next breakpoint.
Wrong guesses of the preorders for the states in $D$-component and the states in $M$ will lead to violation to R1' and R2' in the local consistency test.
With a simple inductive argument we can thus conclude that there can only be one such accepting macrorun.
\end{proof}

\end{document}